\documentclass[11pt, svgnames]{article}

\usepackage{booktabs}
\usepackage{array}
\usepackage[margin=1in]{geometry}
\newcolumntype{C}[1]{>{\centering\arraybackslash}p{#1}}
\usepackage{comment}
\usepackage{cite}

\usepackage{hyperref}
\hypersetup{
colorlinks=true,
linkcolor=blue,
filecolor=blue,
citecolor=blue,  
urlcolor=blue,
linktocpage=true}

\usepackage{amsmath,amssymb,amsthm,mathrsfs,amsfonts,dsfont,amscd,keyval}
\usepackage{mathtools,mathrsfs}
\usepackage{yfonts} 
\usepackage{ytableau}
\usepackage{subfigure, epsfig,graphicx}
\usepackage{braket}
\usepackage{tensor}
\usepackage{bm}
\newcommand{\bvec}[1]{\bm{#1}}

\usepackage{enumerate}
\usepackage{color}
\usepackage{tabularx}
\usepackage{multirow}
\usepackage{float}
\usepackage{tikz}
\usepackage{tikz-cd}
\usepackage{quantikz}
\usepackage{circuitikz}
\usepackage{adjustbox}
\usepackage{pgfplots}
\pgfplotsset{compat=1.12}
\usetikzlibrary{patterns}
\usepackage{quiver} 
\usepackage{ytableau}
\usepackage{nicematrix}
\usepackage{authblk}
\theoremstyle{definition}
\newtheorem{definition}{Definition}
\newtheorem{example}{Example}
\newtheorem*{claim}{Claim}

\theoremstyle{plain}
\newtheorem{theorem}[definition]{Theorem}
\newtheorem{lemma}[definition]{Lemma}
\newtheorem{corollary}[definition]{Corollary}

\newtheorem{assumption}{Assumption}
\theoremstyle{remark}
\newtheorem*{remark}{Remark}
\def\P{\mathcal{P}}

\DeclareMathOperator{\Ima}{im}

\DeclareMathOperator{\supp}{supp}

\DeclareMathOperator{\CNOT}{CNOT}
\newcommand{\CCZ}{\mathrm{C}\mathrm{C}{Z}}

\newcommand{\comments}[1]{}

\begin{document}

\title{No-go theorems for logical gates on product quantum codes}
\author[1]{Esther Xiaozhen Fu}
\author[2]{Han Zheng}
\author[3]{Zimu Li}
\author[3]{Zi-Wen Liu}
\affil[1]{QuICS, University of Maryland}
\affil[2]{Pritzker School of Molecular Engineering and Department of Computer Science, University of Chicago}
\affil[3]{Yau Mathematical Sciences Center, Tsinghua University}

\date{}

\maketitle

\begin{abstract}
    Quantum error-correcting codes are essential to the implementation of fault-tolerant quantum computation. Homological products of classical codes offer a versatile framework for constructing quantum error-correcting codes, especially quantum low-density parity check (qLDPC) codes, with desirable properties. Generalizing the Bravyi--K\"{o}nig topological code argument, we present general algebraic conditions for logical gates which cover codes without geometric locality, and subsequently use them to establish a series of general no-go theorems for fault-tolerant logical gates supported by hypergraph product codes. Specifically, we show that strongly transversal implementations of non-Clifford logical gates are not possible on hypergraph product codes of all product dimensions, and that the dimensions impose various limitations on the accessible level of the Clifford hierarchy gates by constant-depth local circuits. We also discuss examples both with and without geometric locality which attain the Clifford hierarchy bounds. Our results reveal fundamental restrictions on logical gates originating from highly general algebraic structures, extending beyond existing knowledge only in geometrically local, finite logical qubits, transversal, ``robust'', or 2-dimensional product cases. This reveals the algebraic nature of the logical gate problem and provides valuable guidance for the crucial study of fault-tolerant quantum computation with qLDPC codes. 
\end{abstract}

\tableofcontents

\section{Introduction}

While quantum computing is theoretically expected to bring revolutionary advantages over conventional technologies, achieving its full promise will require efficient fault tolerance strategies that guarantee the reliability of the computation at scale. 
Due to the inherent vulnerability of quantum information to diverse noise effects and manipulation inaccuracies, the realization of fault-tolerant quantum computation presents formidable challenges and has spurred intensive efforts across both theoretical and experimental domains.

Utilizing carefully designed encodings with larger physical systems, quantum error correction (QEC)~\cite{shor1995qec,gottesman1997stabilizer} provides a versatile framework for protecting logical quantum information that serves as a cornerstone of fault-tolerant quantum computation~\cite{Shor1996FaultTolerant,Campbell_2017}. 
Crucially, we do not merely want to use QEC codes as memories for preserving static logical quantum information; in order to execute useful computation, we also need to perform logical gates or operations that dynamically process information. 
Therefore, a critical consideration when choosing QEC codes is the encoded logical gates that can be implemented in a fault-tolerant manner---roughly meaning that the errors remain under control.  

An ideal approach for fault-tolerant logical gate implementation is transversal gates, which directly guarantee that errors do not propagate.
Unfortunately, it is dictated by the  Eastin--Knill theorem~\cite{EastinKnill} 
that no nontrivial QEC code
can admit transversal implementations of a universal set of gates. This restriction crystallizes the fundamental tension between information protection and computation in quantum codes that can also be quantitatively refined through information-theoretic lenses~\cite{liuQuantum2022,liuApproximate2023}, indicating that logical gates constitute a crucial but nontrivial feature that generally presents mutual constraints with other important code properties.
In particular, among all quantum gates, the Clifford group gates (which are not universal) are considered the ``cheap'' type with relatively easy protection using stabilizer codes and efficient classical simulation~\cite{gottesman1998heisenberg,PhysRevA.70.052328,nielsen2000quantum}.  
To achieve universal fault-tolerant quantum computation, non-Clifford gates such as $T$ and $\mathrm{C}\mathrm{C}Z$ gates constitute the main bottleneck, thus they are of outstanding importance in the logical gate problem. Notably, in addition to importance at the encoded gate implementation level, QEC codes with transversal non-Clifford gates are essential for magic state distillation~\cite{Bravyi_2012}---a leading fault tolerance strategy that has attracted sustained research interest~\cite{BravyiHaah12,HastingsHaah18,KrishnaTillich18,LBT19,Fang2020PRL,FangLiuPRXQ,wills2024constantoverheadmagicstatedistillation,golowich2024asymptotically,nguyen2024good,lee2024low,golowich2024quantumldpccodestransversal,fang2024surpassingfundamentallimitsdistillation}.

Among possible QEC codes, quantum low-density parity check (qLDPC) codes, i.e.,~families of codes with finite weight and degree stabilizer generators or syndrome measurements, are of particular interest due to their favorability for fault tolerance implementation. 
Considering the hardware limitations of nascent experimental platforms such as superconducting qubits, topological codes with geometrically local structures such as the 2D surface code~\cite{KITAEV20032,bravyi1998quantumcodeslatticeboundary,PhysRevA.86.032324} have long been considered particularly promising for fault tolerance. However, these geometric codes suffer from low encoding rates and thus incur substantial resource overhead for large-scale computation.
By relaxing geometric locality constraints, it is possible to achieve high-rate qLDPC codes, which have recently attracted intense interest from both the theoretical and experimental communities, as they provide a pathway toward efficient fault tolerance with asymptotically constant overhead~\cite{GottesmanLDPC}.


For the construction of QEC, particularly qLDPC codes with desired properties, homological products of chain complexes originating from homological algebra offer a versatile  approach~\cite{Tillich_2014,BH:homological,Audoux2018}. Obtained from the product of classical codes, hypergraph product codes~\cite{Tillich_2014} constitute a highly tractable yet powerful framework that can readily produce qLDPC codes with pretty good code parameters that substantially outperform any geometrically local code and lays the foundation for the surge of progress in qLDPC codes (see, e.g.,~\cite{PRXQuantum.2.040101} for a slightly outdated overview). In particular, the parameters of hypergraph product codes suffice for the constant overhead theorem~\cite{GottesmanLDPC,PhysRevA.87.020304,Leverrier_expander,Fawzi_expander} (further improvements in the code distance require more sophisticated constructions~\cite{fiber,lifted,balanced,pk22,QuantumTanner2022,DHLV,lin2022goodquantumldpccodes}). Meanwhile, the simple structure of hypergraph product codes underpins crucial advantages or convenience in  e.g.~decoding \cite{Leverrier_expander,Fawzi_expander,Grospellier2021combininghardsoft}, syndrome extraction~\cite{GottesmanLDPC,Tremblay22,Manes2025distancepreserving}, logical Clifford gates~\cite{PhysRevX.11.011023,Quintavalle_2023,patra2024targetedcliffordlogicalgates,berthusen_2025}, as well as providing a fertile playground for exploring the interaction between physics and QEC~\cite{tan2024fractonmodelsproductcodes,zhao2024energybarrierhypergraphproduct,rakovszky2024physicsgoodldpccodes}.
Furthermore, recent research highlights their great compatibility with reconfigurable atom array experimental platforms~\cite{Xu_2024,xu2024fastparallelizablelogicalcomputation}.
These features make hypergraph product codes exceptionally important and promising for the actual practical implementation of quantum computation.

In this work, we present a comprehensive study of the possibility of fault-tolerant logical gates, especially non-Clifford gates, on hypergraph product codes of general dimensions. 
It is previously shown by Bravyi and K\"{o}nig that the spatial dimension of topological stabilizer codes limit the Clifford hierarchy level of logical gates that admit constant-depth unitary circuit realizations~\cite{bravyi_classification_2013}. 
Building on the underlying insight, we establish generalized algebraic rules of logical gates applicable to
all quantum code families with or without geometric restrictions, as well as certain useful specialized forms. These conditions provide fundamental lemmas for characterizing fault-tolerant logical gates of general codes. 
We proceed to derive a series of general no-go theorems for fault-tolerant logical gates supported by hypergraph product codes. These no-go theorems extend and refine previous results by Bravyi–König~\cite{bravyi_classification_2013}, which only dealt with geometrically local codes, and Burton–Browne \cite{burton2020}, which considered a special case of 2-dimensional hypergraph product codes satisfying the robustness code property and restricted to a single sector. 
Without imposing any additional code constraint, we are able to show that hypergraph product codes of any product dimension do not support implementations of any non-Clifford gate using tensor product single-qubit gates.  This fills the long-standing gap in understanding why such strongly transversal non-Clifford gates have never been found for hypergraph product codes even when removing geometric locality, despite their widespread study. Moreover, we further generalize the no-go results for hypergraph product codes to the setting of constant depth circuits under a general ``dimension preserving" assumption. In particular, we need a $t$-dimensional hypergraph product code with string-like $X$ or $Z$ logical operator to realize a gate in the $t$-th level of the Clifford hierarchy, and more generally the set of implementable fault-tolerant logical gates is constrained by the quotient of the hyperplane dimensions of the $X$ and $Z$ logical operators. See table \ref{tab:summary of results} for a detailed summary of the no-go results.  


As a key insight, we emphasize that the notion of dimension for hypergraph product code is algebraic, not tied to any notion of geometric locality. Consequently, our results substantially extend previous knowledge on geometrically local codes and illuminate the algebraic nature of the logical gate problem, showing that geometric locality is not the fundamental limiting factor for logical computation, in contrast to many other code properties. 
We demonstrate the tightness of our no-go results by identifying instances of code constructions saturating our upper bounds for Clifford hierarchy of logical gates in situations both with and without geometric locality \cite{Bombin_2007,Bombin_2013,Kubica2015,10.21468/SciPostPhys.14.4.065,Chen_2023,Wang_2024,Breuckmann2024Cups,Lin2024transversal,golowich2024quantumldpccodestransversal}. These ``yes-go" instances further show that the algebraic---rather than geometric---aspect of a code provides a more fundamental and powerful perspective in understanding logical gate properties.  Our theory and results are expected to serve as a foundation for studying even more general types of qLDPC code constructions.

This paper is organized as follows. In Section~\ref{sec:prelim}, we begin by reviewing some key definitions and facts about classical and quantum codes in Section~\ref{subsec: codes}, and logical gates in Section~\ref{subsec: logical gates}. This is followed by a review of the Clifford group and Clifford hierarchy in Section~\ref{subsec: clifford} and the Bravyi--K\"{o}nig theorem in Section~\ref{subsec: bravyi-koenig}. 
In Section~\ref{sec: generalized BK}, we formalize our generalized Bravyi--K\"{o}nig theorems. The new theorem now includes non-geometrically local codes with a simplified version in the case where the logical gate is restricted to the third Clifford hierarchy.   
We review the hypergraph product code framework in Section~\ref{sec: hgp}, introducing the canonical logical representatives as a complete basis. The basic 2-dimensional case is discussed in Section~\ref{sub: 2DHGP} and the most general case in Section~\ref{sub: generalHGP}. Section~\ref{sub: position} introduces necessary definitions for hyperplanes and hypertubes on a hypergraph product code which play key roles in our theory. Section~\ref{sub: alt} discusses key results in classical coding theory that we will use for finding different sets of logical representatives in hypergraph product codes. 
We arrive at the main no-go results on the fault-tolerant logical gates of hypergraph product codes in Sections~\ref{sec:transversal} and \ref{sec:constant_depth}. In Section~\ref{sec:transversal} we treat the case of strongly transversal gates, showing that only Clifford gates can be realized, and in Section~\ref{sec:constant_depth} we treat the more complicated case of constant-depth circuits. In Section~\ref{sec:constant_depth}, our analysis proceeds in three parts: We show that orientation-preserving constant-depth circuits  do not allow non-Clifford gates in Section~\ref{subsection: artificial scenario}, and examine how more refined restrictions on the Clifford hierarchy level of sector-preserving constant-depth circuits  and shape-preserving constant-depth circuits are enforced by the dimensions of $X$ and $Z$ logical representatives, in Section~\ref{subsection: single sector constant depth} and Section~\ref{subsection: constant-depth-general-case}, respectively. 
In Section~\ref{sec: yes-go}, we discuss constructions of fault-tolerant logical gates through various yes-go examples in both geometrically local and geometrically non-local cases. Finally, we conclude with discussion and outlook in Section~\ref{sec: conclusions and outlook}.

\section{Preliminaries}\label{sec:prelim}

In this section, we provide some essential background information about codes and gates, while fixing various notational conventions. 

\subsection{Classical and quantum codes} \label{subsec: codes}
We start with classical codes:
\begin{definition}[Classical code]
	A classical linear error-correcting code $\mathcal{C} \subset \mathbb{F}^n_q$ with parameters $[n,k,d]$ is a $k$-dimensional subspace of the vector space $\mathbb{F}^n_q$ that encodes $k$ logical bits with $n$ physical bits, and has {distance} $d$.
	The Hamming weight of any vector $\bvec{v} \in \mathbb{F}^n_q$, denoted by $|\bvec{v}|$, is defined as the number of nonzero entries in $\bvec{v}$. A code has distance $d$ means that the minimal Hamming weight among all its nonzero codewords is $d$. The code $\mathcal{C}$ is uniquely specified by its \emph{generator matrix} $G$, a $k\times n$ matrix given by $\mathcal{C} = \{\bvec{u}^TG|\bvec{u}\in\mathbb{F}^k_q\}$, or \emph{parity check matrix} $H$, a $(n-k)\times n$ matrix such that $\mathcal{C} = \{\bvec{c}\in\mathbb{F}^n_q|H \bvec{c} = \bvec{0}\}$. 
\end{definition}

The notion of information set will be important in our study. We state a short definition here and will provide more detailed treatments when it is technically used later.
\begin{definition}[Information set]
	Given an \([n, k, d]\) classical code \(\mathcal{C}\), the coordinate positions corresponding to any \(k\) linearly independent columns of \(G\) form an \emph{information set} for \(\mathcal{C}\).
\end{definition}

Then we need several key definitions about quantum codes: 

\begin{definition}[Stabilizer code]
	A quantum code is given by a subspace of the Hilbert space $\mathcal{H}$.
	A stabilizer code $\mathcal{Q}=\mathcal{Q}(\mathcal{S})$ with parameters $[\![n,k,d]\!]$ is a quantum code specified by the joint $+1$ eigenspace of $n-k$ mutually commuting independent Pauli stabilizer generators $\mathcal{S}$, which is a $2^k$-dimensional subspace of the Hilbert space of $n$ qubits, i.e., encodes $k$ logical qubits with $n$ physical qubits, and has {distance} $d$ (formally defined below). 
\end{definition}

\begin{remark}
    For the rest of the paper, we work with qubit CSS codes, but most results immediately generalize to finite field $\mathbb{F}_{2^m}$ with characteristic 2.
\end{remark}

We denote $\langle \mathcal{S}\rangle$ as the group generated by $\mathcal{S}$. (Sometimes $\mathcal{S}$ will also be used interchangeably to denote stabilizer group.) The single Pauli group is given by $\hat{P}=\{\pm 1, \pm i, I,X,Y,Z\}$, and then the $n$-qubit Pauli group is given by $\mathcal{P}_1 = \hat{P}^{\otimes n}$.  $\mathcal{N}(\mathcal{S})$ is the normalizer of $\langle S\rangle$ defined as $ \mathcal{N}(\mathcal{S}) = \{p \in \mathcal{P}_1: \langle \mathcal{S}\rangle p = p\langle \mathcal{S}\rangle \}$.

\begin{definition}[Code distance]\label{def:distance}
	The nontrivial logical operators of $\mathcal{Q}$ are elements of $\mathcal{N}(\mathcal{S})\backslash \langle \mathcal{S}\rangle$. The distance $d$ of a stabilizer code is the minimum weight of any Pauli operator in $\mathcal{N}(\mathcal{S}) \setminus \langle \mathcal{S} \rangle$, i.e., the smallest weight of a nontrivial logical operator. 
\end{definition}

\begin{definition}[Quantum LDPC code]
	A (family of) stabilizer code is called quantum low-density parity-check (qLDPC) if each stabilizer generator acts on a constant number of qubits and each qubit is involved in a constant number of generators (as $n$ grows). 
\end{definition}

\begin{definition}[CSS code and chain complex representation]
	A Calderbank--Shor--Steane (CSS) code is a stabilizer quantum code defined by two classical linear codes $\mathcal{C}_z,\mathcal{C}_x$ with parameters $[n,k_1], [n,k_2]$ respectively, which satisfy the orthogonality condition $\mathcal{C}_x^\perp \subseteq \mathcal{C}_z$. Let $H_z$ be the parity check matrix of $\mathcal{C}_x$ and $H_x$ be the parity check matrix of $\mathcal{C}_z$. The orthogonality condition ensures $H_z H_x^T = 0$. The CSS code is defined by $X$ stabilizers and $Z$ stabilizers given by the row span of $H_x$ and $H_z$ respectively (with the commutativity of the stabilizer group guaranteed by the orthogonality condition), and has parameters $[\![n,k_1 +k_2 -n,d]\!]$. It corresponds to a three-term chain complex: 
	\[
	\mathbb{F}_2^{m_z} \xrightarrow{\partial_z^T} \mathbb{F}_2^n \xrightarrow{\partial_x} \mathbb{F}_2^{m_x},
	\]
	where \( \partial_x\) gives the parity check matrix $H_x$ and \( \partial_z \) gives the parity check matrix $H_z$, with the orthogonality condition directly enforced by the boundary map rule $\partial_z \partial_x^T = 0$. 
\end{definition}

The logical $X$ and $Z$ operators are given respectively by 
\begin{align}
	L_x = \ker{H_z}/\mathrm{im} \;H_x^T, \quad L_z = \ker{H_x}/\mathrm{im}\; H_z^T,
\end{align}
or equivalently,
$\mathcal{C}_z/\mathcal{C}_x^\perp, \mathcal{C}_x/\mathcal{C}_z^\perp$.
The code distance is simply given by
\begin{align}
	d = \min\{d_x, d_z\},
\end{align}
 $ d_x = \min \{|L_x|\}$ and $d_z = \min \{|L_z|\}$, where $|P|$ denotes the weight of the non-identity elements of the operator $P$.

By Definition~\ref{def:distance}, a logical operator is trivial if it is generated by the stabilizers given by $H_x$ and $H_z$. Two logical representatives $L_1,L_2$ are equivalent, denoted as $L_1 \sim L_2$, if they differ by a stabilizer.

\subsection{Logical gates}\label{subsec: logical gates}

The theme of this work is to understand fault-tolerant logical gates compatible with quantum codes. 
Here we formally introduce several definitions and notations regarding logical gates which will be used for the rest of the paper.

\begin{definition}[Support]
	The \textit{support} of an operator $P$, denoted $\mathrm{supp}[P]$, is the set of physical qubits on which $P$ acts nontrivially (i.e., not as the identity). 
	
	The \textit{size} of the support, denoted $|P|$, is the weight of the non-identity elements of the operator $P$.
\end{definition}

\begin{definition}[Code projector]
	Let $\Pi$ denote the projector from the Hilbert space $\mathcal{H}$ onto the code space of $\mathcal{Q}$. For any encoded state $\rho$, we have $\Pi \rho = \rho \Pi$. 
\end{definition}

\begin{definition}[Logical]
	A unitary operator $U$ is a logical gate/operator for the code $\mathcal{Q}$ if it preserves the code space of $\mathcal{Q}$, i.e., $U\Pi U^\dagger = \Pi$.
\end{definition}

In this work, we restrict ourselves to logical gates that preserves the code space. That is, we do not consider the extended notion of logical gate where $U\Pi_1 U^\dagger =\Pi_2$ as morphisms between two codes with $\Pi_1 \neq \Pi_2$ \cite{bravyi_classification_2013}.

Loosely speaking, logical gates that are sufficiently ``local'', so that the error propagation is well confined and thus the output of the circuit remains correctable if the physical error rate is low enough, are fault-tolerant.  The most well-studied example of such gates are transversal gates:
\begin{definition}[Transversal gate]\label{def: transversal}
	A unitary operator \( U \) is said to be \textit{transversal} if it takes the form \( U = \bigotimes_{i=1}^n u_i \), where each \( u_i \) is a single qubit unitary.
	
\end{definition} 
Note that here, for technical reasons, we adopt a \textit{strong} notion of \emph{transversality} that requires a gate to decompose as a tensor product of single-qubit unitaries, which is more stringent than the other usual definition of transversal gate as a quantum operation that does not propagate errors within the same code block (but can spread a single error to at most a single error in a different code block). That is, the strong definition excludes the more general block-wise transversality (which is however treated as constant-depth circuits in our context), and this distinction is important for our results in Section~\ref{sec:transversal} and \ref{sec:constant_depth}.  

The transversality property is ideal but not necessary for fault tolerance. For infinite code families, we can make the following relaxation to locality-preserving maps, which still guarantee that a local error can only be propagated to a bounded number of other qubits and thus naturally remain fault-tolerant:
\begin{definition}[Constant-depth circuit gate]\label{def: constant-depth}
	A unitary operator \( U \) is said to be implementable by a \textit{constant-depth circuit} if it can be decomposed into a local quantum circuit of $O(1)$ depth, i.e., $O(1)$ layers of disjoint local (not necessarily geometrically local) gates which acts only on a constant number of qubits, independent of code length $n$. The support of any operator $P$ under the action of $U$ satisfies \( |UPU^\dagger| = c \cdot |P| \), where \( c = O(1) \) is a constant. 
\end{definition}

For instance, the [\![15,1,3]\!] quantum Reed-Muller code supports a transversal \( T \) gate. In contrast, a $\mathrm{CC}Z$ gate implemented on three copies of the toric code is \textit{not} a transversal gate under our definition, but falls under the case of gates implementable by constant-depth circuits.

\subsection{Clifford group and Clifford hierarchy}
\label{subsec: clifford}

As established, Clifford group gates and the extended notion of Clifford hierarchy play pivotal roles in quantum computing. A driving goal of our work is to systematically understand the interplay between the algebraic structures and logical gates of stabilizer codes, with the Clifford hierarchy providing the essential framework. 

The Clifford hierarchy is a nested tower of sets of unitary operators with original motivations associated with gate teleportation or injection~\cite{GottesmanChuang_1999}.
  The levels of the Clifford hierarchy can be defined recursively as follows.
\begin{definition}[Clifford hierarchy]
Denote the $k$-th level of the Clifford hierarchy as $\P_k$. Define the zeroth level $\P_0$ to be the phase group $\{\pm 1, \pm  i\}$, and the first level $ \P_{1} $ to be the Pauli group given by $\mathcal{P}_1 = \hat{P}^{\otimes n}$, $\hat{P}=\{\pm 1, \pm i, I,X,Y,Z\}$. Then the second level $ \P_{2} $ is given by the group of automorphisms of the Pauli group: $ \{U\mid UPU^{\dagger}\in \P_1, \forall P \in \mathcal{P}_1\} $ and is called the Clifford group. Higher levels of the hierarchy are defined recursively as
\begin{equation}\label{eq:cliffordhierarchy}
 	\P_{k} := \{U\mid UPU^{\dagger}P^{\dagger}\in \P_{k-1}, \forall P \in \mathcal{P}_1\}.
 \end{equation} 
\end{definition}

The Clifford group $ \P_{2} $ includes Hadamard gate, Pauli gates, phase gate $S = Z^{1/2}$, $\CNOT$  gate (and their compositions).
For $l >2$, $\mathcal{P}_l$ is not closed and not a group. 
Two particularly important and widely used non-Clifford gates, $T = Z^{1/4}$ and $\mathrm{C}\mathrm{C}Z$, belong to $\P_3$. More generally, $Z^{1/2^{l-1}}$ and $\mathrm{C}^{l-1}Z$ gates are typical examples of gates in $\P_l$ for $l\geq 2$.

\subsection{Correctability, logical operators, and the Bravyi--K\"{o}nig theorem}\label{subsec: bravyi-koenig}

In this section, as the foundation for our later results, we formally review various basic concepts and techniques concerning (and connecting) error correction and logical operators, leading up to the Bravyi--K\"{o}nig theorem.

\begin{theorem}[Knill--Laflamme~\cite{KL}]\label{thm:Knill--Laflamme}
Let \( \mathcal{Q} \) be a quantum code, and let \( \Pi \) denote the projector onto the code space $\mathcal{C}$. Suppose \( \mathcal{N} \) is a quantum channel with Kraus operators \( \{ E_\alpha \} \). A necessary and sufficient condition for the existence of a recovery channel \( \mathcal{R} \) (completely positive trace-preserving map) that corrects $\mathcal{N}$, i.e.,
\[
\mathcal{R} \circ \mathcal{N} (\rho) = \rho, \quad \forall\rho \in \mathcal{C},
\]
is that
\begin{equation} \label{eq: Knill laflamme}
    \Pi E_\alpha^{\dagger} E_\beta \Pi = c_{\alpha \beta} \Pi
\end{equation}
for some complex number \( c_{\alpha \beta} \), for all $\alpha,\beta$.
\end{theorem}

\begin{lemma}[Cleaning Lemma \cite{Bravyi_2009}]\label{lemma: correctability-stabilizer-code} 
Let \( \mathcal{Q} \) be a stabilizer code defined on physical qubits \( \Lambda \). Given a subset \( R \subset \Lambda \),   one and only one of the following holds:
\begin{enumerate}[(i)]
    \item There exists a non-trivial logical operator \( L \in \mathcal{N}(\mathcal{S}) \setminus \langle \mathcal{S} \rangle \) whose support is entirely contained in \( R \);
    \item For every logical operator \( L \in \mathcal{N}(\mathcal{S}) \setminus \langle \mathcal{S} \rangle \), there exists a stabilizer \( s \in \mathcal{S} \) such that \( L\cdot s \) acts trivially on  \( R \).
\end{enumerate}
We say that a subset $R$ is \emph{cleanable} if (ii) holds. 
\end{lemma}

In other words, if the Knill--Laflamme condition is satisfied for any set of errors constrained to $R$, we know that errors locally supported on $R$ do not alter the logical information and $R$ cannot fully support a logical operator. By the Cleaning Lemma, $R$ is cleanable, i.e.,~in Lemma~\ref{lemma: correctability-stabilizer-code} (i) does not hold and (ii) holds. This is also equivalent to saying that $R$ is a \textit{correctable region}:

\begin{definition}[Correctable region] \label{def:correctable}
A subset of qubits $R$ is \emph{correctable} if and only if there exists a recovery channel $\mathcal{R}$ that corrects the erasure of all qubits in $R$; that is, 
\begin{align}
\mathcal{R}(\mathrm{Tr}_R\rho) = \rho, 
\end{align}
for
any encoded state $\rho \in \mathcal{C}$.
\end{definition}

\begin{lemma}\label{lemma: cleaning-lemma-correctable}
Given a stabilizer code $\mathcal{Q}$, a region $R$ is called \emph{correctable} if and only if region $R$ is cleanable. 
\end{lemma}

\begin{remark}
In our later proofs, we will always conclude with a logical operator $K$ defined on a correctable region $R$. Without further assumption, we can expand $K = \sum_i P_i$ as a sum of Pauli operators, all of which are supported on $R$. Then using the same argument as the Knill--Laflamme condition:
	\begin{align}
		K {\ket{\psi}} = \Pi K \Pi {\ket{\psi}} = \sum_i \Pi P_i \Pi {\ket{\psi}} = \sum_i c_i {\ket{\psi}}.
	\end{align}
This is independent of the encoded state $\ket{\psi}$, which is an encoded logical state, verifying that $K$ is trivial on the code space.
\end{remark}

The next result indicates that the combinations of smaller correctable regions are still correctable: 
\begin{theorem}[Union Lemma \cite{Bravyi_2009,pastawski_fault-tolerant_2015}]
Given a stabilizer code, let $R_1$ and $R_2$ be two disjoint sets of qubits. Suppose there exists a complete set of stabilizer generators $\mathcal{S}$ such that the support of each generator overlaps with at most one of $\{R_1,R_2\}$. If $R_1$ and $R_2$ are correctable, then the union $R_1 \cup R_2$ is also correctable. 	
\end{theorem}

With these concepts and techniques, Bravyi and K\"{o}nig proved the following theorem which indicates topological restrictions on the logical operators supported by topological stabilizer codes.

\begin{theorem}[Bravyi--König \cite{bravyi_classification_2013}]
    Suppose a unitary operator $U$ implementable by a constant-depth quantum circuit preserves the code space $\mathcal{C}$ of a topological stabilizer code given by geometrically local stabilizers on a $D$-dimensional lattice, $D \geq 2$. {For sufficiently large  lattices and  code distance,} the restriction of $U$ onto $\mathcal{C}$ implements an encoded gate from the set in the $D$-th level of Clifford hierarchy.  
\end{theorem} 
See \cite{bravyi_classification_2013} for further details. Here we give a sketch of the proof to distill the core reasoning and insights. 
Hereafter we will frequently use the following notion of operator commutator:
\begin{definition}[Group commutator]
Given two unitary operators $P, Q$,  their \emph{group commutator} is defined as
   \begin{equation}
    [P,Q] \coloneqq PQP^\dagger Q^\dagger.
\end{equation} 
\end{definition}

\begin{proof}  
(Sketch) Here we consider a stabilizer code given by stabilizer generators that are geometrically local, i.e., act on a constant number of qubits within a ball of constant radius, in $D$ spatial dimensions. Further, let $U$ be a logical gate implementable by a constant-depth circuit that respects the geometric locality of the code, i.e., under the evolution $U$, $O \to UOU^\dagger$ enlarges the support by at most a constant radius. The goal is to show that $U$ must be in $\mathcal{P}_D$ for all large enough codes. 

\paragraph{Correctable partitioning} The proof relies on the possibility of  partitioning the physical qubits into only $D+1$ subsets that are individually correctable. This can be achieved by a clever spatial partitioning of the $D$-dimensional lattice into $D+1$ groups of disjoint subregions in a way that (a constant size expansion of) the subregions within each group are i) sufficiently small to be correctable, and ii) sufficiently far separated to ensure that no pair of subregions have non-trivial support for the same stabilizer generator (see \cite{bravyi_classification_2013} for an illustration).
Then using the Union Lemma one can show that each of these $D+1$ groups, each consisting of many disjoint but correctable subregions, is also correctable. 
We label these regions as $R_1,\cdots, R_{D+1}$.
    
\paragraph{Deforming logical representatives} 
The next step involves the Cleaning Lemma which allows us to obtain different logical representatives through deforming each representative from a given cleanable region. Specifically, one can obtain $D$ different logical representatives, labeled $\{L_1,\cdots L_D\}$ such that $L_i$ is cleaned from the region $\Lambda_i$. Note that for the constant-depth case, special care has to be taken in choosing the partitions so that a logical representative can be cleaned from the region $\mathcal{B}(R_i)$ defined as the maximum support of $UOU^\dagger$, for all $O$ that is constrained to $R_i$ and for all constant-depth $U$.

\paragraph{Group commutator}  
A key insight for the argument to work is that one can bound the support of the group commutator $K_2 = [UL_1U^\dagger, L_2]$ to the intersection of the support of $L_1$ with $L_2$. 
Hence, by recursively building group commutators $K_i = [UP_iU^\dagger,P_{i+1}]$, we obtain logical operator $K_s$ with smaller and smaller support, with $K_s$ having trivial support on $\cup_{i=1}^{s+1}\Lambda_i$. In particular, $K_D$ is supported only on $\Lambda_{D+1}$, which by construction is a correctable region. 

\paragraph{Clifford hierarchy restriction}
Then the puzzle ``What logical operators can be supported on a correctable region?'' can be resolved with the following fact.
\begin{lemma}\label{lemma: phase for trivial operator}
Let $K$ be a logical operator of the form of a group commutator $K = UPU^\dagger P^\dagger$, where $P \in \P_2$ is a logical Pauli operator and $U$ is a unitary operator. Suppose $K$ is supported on a correctable region, then $K\Pi = \pm \Pi$.
\end{lemma}
\begin{proof}
    Since $K$ is supported on a correctable region and is a logical operator, then we know that $K$ must act trivially on the codespace, i.e. $K\Pi = c\Pi$ with $|c|=1$. Multiplying $P$ on both sides, we obtain $UPU^\dagger \Pi = cP\Pi$. 

    Then, using the fact that $P^2 = \mathbb{I}$, we can square both sides of the equation to eliminate $P$ from the equality:
    \begin{align}
        & (UPU^\dagger)^2= (cP)^2\\
        &LHS = UP^2 U^\dagger = \mathbb{I}\\
        &RHS = (cP)^2 = c^2\mathbb{I} \\
        &\implies 1= c^2 
    \end{align}
    which implies that $c = \pm 1$. 
\end{proof}
By (\ref{eq:cliffordhierarchy}) we have $K_D\in \P_1$, and by induction, one can show that $K_{D-j} \in \P_{j+1}$, so $K_{2} \in \P_{D-1}$ which implies that $U \in \P_D$. 
\end{proof}

\section{Generalized Bravyi--König theorems}\label{sec: generalized BK}
In the previous section, we reviewed the Bravyi--K\"{o}nig theorem, which indicates a tradeoff between fault-tolerant logical gates and the spatial dimension of geometrically local stabilizer codes. From this section onward, we move beyond geometrically local codes and seek to understand the limitations on fault-tolerant logical gates in more general code families.

As a step towards this, we first generalize the original Bravyi--K\"{o}nig argument to algebraic forms that encompass all quantum code families without any restriction on e.g.~spatial geometry and logical qubits. The key insight is that if we are able to gain further understanding of the structure of different sets of logical representatives of a code family, then we can extend the logic of Bravyi--K\"{o}nig to establish bounds on the Clifford hierarchy level even when the code is not geometrically local or even not a qLDPC code. 
With this generalization, the remaining task is to examine specific code constructions and attempt to construct logical representatives in an adversarial way so that one can obtain $K_i$ on a correctable region for the least $i\in \mathbb{Z}$. This leads to fundamental restrictions on logical gates extending beyond geometrically local codes, which will be important as we increasingly shift gear towards non-local codes.

We now rigorously formalize the generalized Bravyi--König theorem below as a universal lemma.
Note that this can also be generalized to the case where $U \Pi_1 = \Pi_2 U$ with $\Pi_i$ being projectors of different stabilizer codes.

\begin{lemma}[Generalized Bravyi--König theorem]
\label{theorem: generalize bk}
Let $\mathcal{Q}$ be a $[\![n,k,d]\!]$ stabilizer code. Let $U$ be a constant-depth circuit that preserves the code space. Suppose for any sequence $\{\overline{P}_j\}_j$ where $\overline{P}_j$ is a logical operator from the set generated by $2k$ X and Z logical operators $\langle \overline{L}_{x,i}, \overline{L}_{z,i}\rangle , i\in[k]$, there exist a corresponding sequence of logical representatives $\{{P}_j\}_j$, where $P_j$ is a logical representative for $\overline{P}_{j}$, such that for the following sequence of operators: 
\begin{align*}
K_1 &= UP_1 U^\dagger, \\
K_2 &= [K_1, P_2],\\
&\cdots\\
K_j &= [K_{j-1},P_j],
\end{align*}
$K_j$ is supported on a correctable region, that is, $K_j = \pm \mathbb{I}$. Then, $U$ is a logical gate constrained to $\P_j$. 
\end{lemma}
\begin{proof}
We prove this by induction on ${K}_j$ using the definition of $\mathcal{P}_k$ in Eq.~(\ref{eq:cliffordhierarchy}).
If $j=2$, then ${K}_2$ is a trivial logical operator by Lemma~\ref{lemma: phase for trivial operator}, so we conclude that $K_2 = \pm 1$ which implies that $U{P}_1 U^\dagger \in \mathcal{P}_1$. By Eq.~(\ref{eq:cliffordhierarchy}), $U \in \mathcal{P}_2$ is a Clifford gate. 

For $j>2$, suppose ${K}_{j-m} \in \mathcal{P}_m$, then ${K}_{j-m} = [{K}_{j-m-1},{P}_{j-m}]$ for all $ {P}_{j-m} \in \langle {L}_{x,i}, {L}_{z,i}\rangle_{i \in[k]}$ which implies that ${K}_{j-m-1} \in \mathcal{P}_{m+1}$. Recursively, we obtain ${K}_1 \in \mathcal{P}_{j-1}$, which implies that $U \in \mathcal{P}_j$. 
\end{proof}

We now state two useful results about supports which are formulated for the transversal and constant-depth circuit types of fault-tolerant logical gates  respectively.

\begin{lemma}\label{lemma: constant depth support}
    Let $K = [UPU^\dagger, Q] = UPU^\dagger Q (UPU^\dagger)^\dagger Q^\dagger$. If $U$ is a constant-depth circuit and $P$, $Q$ are logical Pauli operators, then the support of $K$ is constrained to 
    \begin{align}\label{eq: constant depth support}
        \mathrm{supp}\left[UPU^\dagger V(UPU^\dagger)^\dagger\right] \cup \mathrm{supp}\left[V\right],
    \end{align}
    where $V = Q|_{\supp[UPU^\dagger]}$, so that $\supp[V] = \supp[Q]\cap \supp[UPU^\dagger]$.
\end{lemma}

\begin{proof}
$UPU^\dagger$ acts nontrivially only on the qubits that is in the support of \( UPU^\dagger\). Thus, when we conjugate \( Q \) by \( UPU^\dagger \), only the non-trivial component of \( Q \) that overlap with the support of \( UPU^\dagger \) contribute nontrivially. Therefore, the term
\begin{equation}
UPU^\dagger Q (UPU^\dagger)^\dagger
\end{equation}
can be decomposed as having contribution from two parts:
\begin{equation}\label{eq:supp of const depth}
\left(UPU^\dagger V (UPU^\dagger)^\dagger\right) \cdot (V^{\dagger}Q).
\end{equation}
The support of the first term is contained in the region given by 
\[
\mathrm{supp}\left[UPU^\dagger V(UPU^\dagger)^\dagger\right],
\]
while the support of the second term is given by $\supp[V^\dagger Q ]$.
The product of (\ref{eq:supp of const depth}) with $Q^\dagger$ then gives:
\begin{equation}
UPU^\dagger V (UPU^\dagger)^\dagger \cdot V^{\dagger},
\end{equation}
The support of which is given by
\[
\mathrm{supp}\left[UPU^\dagger V(UPU^\dagger)^\dagger\right]\cup \mathrm{supp}[V^\dagger].
\]
This concludes the proof.
\end{proof}
When $U$ is a transversal gate, the task of finding sets of logical representatives to obtain a no-go result can be reduced to the following statement:
\begin{corollary}
Let \( U \) be a transversal unitary gate. Suppose there exists a set of logical representatives \( \{P_j\}_j \) where each \( P_j \) is a representative of the logical operator \( \overline{P}_j \) such that the intersection of their supports,
\[
\operatorname{supp}[P_1] \cap \operatorname{supp}[P_2] \cap \cdots \cap \operatorname{supp}[P_j],
\]
is a correctable region. Then the logical action of \( U \) is restricted to $\mathcal{P}_j$.
\end{corollary}

\begin{proof}
When $U$ is a transversal gate, (\ref{eq: constant depth support})  reduces to the following: 
    \begin{align}\label{eq: constant depth support2}
\mathrm{supp}\left[UPU^\dagger V (UPU^\dagger)^\dagger\right] \cup \mathrm{supp}\left[ V\right]
= \mathrm{supp}[V]
= \supp[Q] \cap \supp[P].
\end{align}
The support of $K_2$ is given by $\supp[P_2]\cap\supp[P_1]$, and the support of $K_3$ is then given by $\supp[P_3]\cap (\supp[P_1]\cap \supp[P_2])= \supp[P_3]\cap\supp[P_2]\cap\supp[P_1]$.
Doing this recursively yields $K_j = \operatorname{supp}[P_1] \cap \operatorname{supp}[P_2] \cap \cdots \cap \operatorname{supp}[P_j]$. Then, this implies that $K_j$ is a correctable region. Following the result from Lemma~\ref{theorem: generalize bk}, we conclude that $U \in \P_j$.  
\end{proof}

We end this section with a brief discussion of the connection between our understanding of the Clifford hierarchy with the number of logical qubits required in the proof of the Bravyi--König theorem (as in Lemma~\ref{theorem: generalize bk}). It will not be essential to the results in the subsequent sections, but motivates the treatment of logical gate limitations for three separate cases where we respectively consider (i) all the logical qubits as per the statement in Lemma~\ref{theorem: generalize bk}; (ii) only a finite number of logical qubits is involved for each element in the sequence $\{\overline{P}_j\}_j$, so that the total number of logical qubits scales linearly with the level of the Clifford hierarchy; (iii) only $O(1)$ number of logical qubits on which we can obtain the limitations of the logical gates.

There are several reasons why we treat these three cases separately. Firstly, for commonly considered logical gates including $T, \mathrm{CC}Z$ and $\mathrm{C}Z$ etc., there always exist a sequence of logical operators of a single logical qubit $\overline{P}_j$ to sequentially lower the level of Clifford hierarchy by using group commutator. So this motivates a theoretical understanding of no-go results specific to Case (ii), where we only need a sequence involving a few logical qubits each time. Case (ii) is interesting for known subgroups such as the diagonal Clifford gates which have been fully characterized in~\cite{PhysRevA.95.012329}. Secondly, we are often interested in logical gates acting on a finite number of logical qubits, such as in the consideration of addressable gates  \cite{he2025addressable,Quintavalle_2023}, and there are also examples of codes with only a constant number of logical qubits, so Case (iii) will be useful for such instances. Therefore, it is also a well-motivated case that should be separately checked.  

Theoretically speaking, we cannot rule out the possibility that a logical operator realizes its full level in the Clifford hierarchy only when we study its action on all the $k$ logical qubits in the code. Thus, Case (i) is required to understand the limitations on logical gates in its full generality.

Remarkably, it turns out that it is sufficient to consider just Case (iii) if we are able to show that the logical gates are restricted to $\P_3$ (which encompasses the most important gates including $T,  \mathrm{CC}Z, \mathrm{C}S$ and Clifford gates), where we only need to restrict our attention to 3 logical qubits.

\begin{lemma}[Generalized Bravyi--König theorem II]\label{lemma: clifford-hierachy-basis}
Let $U$ be a constant-depth circuit that preserves the code space. Suppose for any sequence $\{\overline{P}_j\}_j$, $j\leq3$, where $\overline{P}_j$ is a logical operator of {a single logical qubit} from the set generated by $2k$ X and Z logical operators $\langle \overline{L}_{x,i}, \overline{L}_{z,i}\rangle, i\in[k]$, there exist a corresponding sequence of logical representatives $\{{P}_j\}_j$, where $P_j$ is a logical representative for $\overline{P}_{j}$, such that for the following sequence of operators: 
\begin{align*}
K_1 &= UP_1 U^\dagger, \\
K_2 &= [K_1, P_2],\\
&\cdots\\
K_j &= [K_{j-1},P_j],
\end{align*}
we have that $K_j$ is supported on a correctable region. Then, $U$ is a logical gate constrained to $\P_j$.
\end{lemma}

\begin{proof}
We will show that the result does not hold for $j=4$ from which the proof for why $j<4$ cases work can be easily deduced.

We start with $K_4 = \pm 1$ and work our way back to $K_1$. Let $\mathcal{B}$ denote the basis $\{\overline{L}_{x,i},\overline{L}_{z,i}\}$. $K_4 = K_3 P_4 K_3^\dagger P_4^\dagger$ for any basis element $\overline{P}_4 \in \mathcal{B}$. Then, all Pauli operators in $\P_1$ must either commute or anti-commute with $K_3$, so $[K_3,P]=\pm1$ holds for all elements in $\P_1$, indicating that $K_3$ must be in $\P_1$.

Next, note that $K_3=[K_2,P_3]$ for all $\overline{P}_3 \in \mathcal{B}$, and $K_3 \in \P_1$, which implies that $K_2 P_3 K_2^\dagger \in \P_1$. Since the Pauli group is closed under multiplication, $K_2 P K_2^\dagger \in \P_1$ for all $P \in \P_1$. Therefore, we conclude that $K_2 \in \P_2$, i.e., in the Clifford group. 

For $K_2=[K_1,P_2]$ where $K_2 \in \P_2$, we use the fact that Clifford group is also closed under multiplication to conclude that $K_2 = [K_1,P] \in \P_2$ for all $P \in \P_1$. This indicates that $K_1$ must be in $\P_3$. 

Finally, $K_1 = UP_1U^\dagger$ is in $\P_3$. However, it is possible that there are some elements in $\P_1$ such that $UPU^\dagger$ is in some higher level of the Clifford hierarchy, even though we know that for all basis element $Q \in \mathcal{B}$ $UQU^\dagger \in \P_3$. As an aside, we can show that there must be a basis element $Q'$ such that $UQ'U^\dagger\in \P_3 -\P_2$, using the fact that $\P_2$ is closed. Therefore, this lemma only works for $j \leq 3$. 
\end{proof}

From the proof above, we show that up to logical gates in $\P_3$, Case (iii) is a sufficient condition for Case (i), but we are unable to show this for higher level gates due to the fact that $\P_n$ is not closed for $n\geq 3$.

An interesting fact we found from the proof is that for any logical gate in $\P_4$, we can sequentially lower the  Clifford hierarchy level each time by using a group commutator and a Pauli operator for a single logical qubit. One may then naturally wonder if for a logical gate in a higher level of the Clifford hierarchy, it is possible to use only a constant number of logical qubits for each round of group commutator to lower the level by one. If this is true, then we can potentially establish a simpler statement for Lemma~\ref{theorem: generalize bk} and it will substantially lower the search complexity of determining the level of Clifford hierarchy for any logical gate and also make it easier to obtain limitations of logical gates for a given code. This will not be essential to the present work and we leave its further investigation for future research.

We briefly remark that we implicitly require that any logical gate implementable by a constant-depth circuit lies in the Clifford hierarchy. We consider this requirement natural and mild: it has been shown to be true for transversal gates~\cite{disjointness_2018} and no reasons or instances of logical gates on stabilizer codes that do not belong to the
Clifford hierarchy are known.
Finally, note that given our knowledge about the logical gates, we cannot entirely rule out the possibility of a logical gate whose level in the Clifford hierarchy is only maximized when acting on a large number of logical qubits. In the subsequent proofs, we will always clarify if the proof works for a subset of logical qubits for $\{\overline{P}_j\}_j$, or if it holds for all logical qubits of a given code as well as any other assumptions that we make. 

\section{{On hypergraph product codes}}\label{sec: hgp}

Hypergraph product codes are quantum codes corresponding to chain complexes obtained from taking homological/tensor products of 1-complexes which represent classical codes. Here the number of seed complexes is referred to as the \emph{dimension} of the product, but note that this should not be confused with the common spatial notion of dimensions, for which we use the notation ``xD''. 

The standard and most understood notion of hypergraph product codes is simply the 2-dimensional ones obtained from the product of two classical codes~\cite{Tillich_2014}, but they can also be generalized to define higher-dimensional hypergraph product codes in analogy to topological codes in higher spatial dimensions~\cite{Zeng_2019}. Here we provide formal definitions and discuss various key properties and techniques.

\subsection{Two-dimensional hypergraph product codes}\label{sub: 2DHGP}    

We start with the standard 2-dimensional hypergraph product codes. A more comprehensive introduction can be found in \cite{Tillich_2014,BH:homological,Leverrier_2015,burton2020,Quintavalle_2023,golowich2024quantumldpccodestransversal}. 

\begin{definition}[2-dimensional hypergraph product code]
Let $\mathcal{A} = \partial_a: \mathbb{F}_2^{n_a} \xrightarrow{\partial_a} \mathbb{F}_2^{m_a}$ and $\mathcal{B} = \partial_b: \mathbb{F}_2^{n_b} \xrightarrow[]{\partial_b} \mathbb{F}_2^{m_b}$ be two 1-complexes. A 2-dimensional hypergraph product code can be obtained from their tensor product $\mathcal{A} \otimes \mathcal{B}$ as shown below. 
\[
\begin{array}{ccc}
\mathbb{F}_2^{n_a} \otimes \mathbb{F}_2^{n_b} & \xrightarrow{I_{n_a} \otimes \partial_b} & \mathbb{F}_2^{n_a} \otimes \mathbb{F}_2^{m_b} \\
\Big\downarrow{\scriptstyle \partial_a \otimes I_{n_b}} &  & \Big\downarrow{\scriptstyle \partial_a \otimes I_{m_b}} \\
\mathbb{F}_2^{m_a} \otimes \mathbb{F}_2^{n_b} & \xrightarrow{I_{m_a} \otimes \partial_b} & \mathbb{F}_2^{m_a} \otimes \mathbb{F}_2^{m_b}
\end{array}
\]
\end{definition}
Its code length $n$ is given by $m_a n_b + n_a m_b$. The physical qubits are split into two sectors with the first one having $m_a n_b$ physical qubits and the second one having $n_a m_b$ qubits. Each of them can be depicted as a 2-dim $m_a \times n_b$ ($n_a \times m_b$) lattice. 

The parity check matrices are given by:
\begin{align}\label{eq:2dHGP_XZ}
H_x = \begin{pmatrix} I_{m_a}\otimes \partial_b \quad \partial_a \otimes I_{m_b} \end{pmatrix}, \quad 
H_z = \begin{pmatrix} \partial_a^T \otimes I_{n_b} \quad I_{n_a} \otimes \partial_b^T \end{pmatrix}.
\end{align}
It is straightforward to verify that the orthogonality condition $H_x H_z^T = 0$ holds, i.e.,
\begin{align}  
	\begin{pmatrix} I_{m_a}\otimes \partial_b \quad \partial_a \otimes I_{m_b} \end{pmatrix}
	\begin{pmatrix} \partial_a \otimes I_{n_b} \\ I_{n_a} \otimes \partial_b \end{pmatrix} = 0,
\end{align}
confirming that this definition yields a valid CSS code.

Evidently, only if the seed classical codes are LDPC then the resulting quantum code is LDPC (So in general a hypergraph product code need not be a qLDPC code).

The $Z$ logicals are given by $\ker{H_X}/\Ima H_Z^T$ and the $X$ logicals are given by $\ker H_Z /\Ima H_X^T$. As it turns out, the idea of punctures can give a more explicit description of the logical representatives. Here, we adopt the following definitions from \cite{burton2020}, and use punctures to describe the logical representatives of hypergraph product codes.

\begin{definition}[Puncture]
	Let $M: \mathbb{F}_2^n \rightarrow \mathbb{F}_2^m$ be a binary linear map. We say a set $\gamma \subset [n]$ punctures $M$ if no vector from the row span of $M$ has support completely contained in $\gamma$. 
	Alternatively, we say that $\gamma$ is $|\gamma|$-\emph{puncturable} to $M$.
\end{definition}

If $\gamma$ is a set with the maximal number of punctures, then the subspace $\mathbb{F}_2^{\gamma} \subset \mathbb{F}_2^n$ has to be a complement of the row span. As complements are not unique, it is possible to choose different maximal subsets of punctures, which is a crucial observation in proving our main results later. 

\begin{definition}[Standard form] For any $k$-dimensional linear code over $\mathbb{F}^n_2$, generator matrix $G$ and parity check matrix $H$ can take the standard form
	\begin{align}
		G & =\left(\begin{array}{ll}
			I_k & J^T
		\end{array}\right), \\
		H & =\left(\begin{array}{ll}
			J & I_{n-k}
		\end{array}\right),
	\end{align}
	where $J \in \mathbb{F}^{(n-k) \times k}_2$.
\end{definition}
\begin{example}
	Once written in the standard form, it is quite obvious to read off a maximal set of punctures: $\gamma = \{1, 2, \cdots, k\}$ is a maximal $k$-puncture to $H$. As mentioned before, punctures might not be unique. 
\end{example}

Let
\begin{align}
	\ker \partial_a = k_a, & \quad \ker \partial_b = k_b, \\
	\ker \partial_a^T = k_a^T, & \quad \ker \partial_b^T = k_b^T.
\end{align}
Putting $\ker \partial_a, \ker \partial_a^T, \ker \partial_b$ and $\ker \partial_b^T$ into standard form. Then, one can also read off that the number of punctures of $\partial_a^T$ and $\partial_b^T$ is equal to $k_a^T$ and $k_b^T$ respectively. 

First fix a set of linearly independent vectors from the standard form of $\ker \partial_a$ as $\{ \xi^{(1)}_a, \cdots, \xi^{(k_a)}_a\} $, with $\xi_a^{(i)}\in\mathbb{F}_2^{n_a}$. Its associated \emph{pivots} are defined as the unit vectors (or one-hot vectors) $\{ e^{(1)}_a, \cdots, e^{(k_a)}_a \}$, $ e_a^{(i)}\in\mathbb{F}_2^{n_a}$, corresponding to the first $k_a$ columns of $\partial_a$ in standard form. We also fix a basis from $\ker \partial^T_a $ as $ \{ \zeta^{(1)}_a, \cdots, \zeta^{(k^T_a)}_a\}$, $\zeta_a^{(i)}\in\mathbb{F}_2^{m_a}$, with its pivots $\{ f^{(1)}_a, \cdots, f^{(k^T_a)}_a \}$ as the unit vectors corresponding to the first $k_a^T$ columns of $\partial_a^T$ in standard form. Similarly we define the linearly independent vectors and their associated pivots for $\partial_b$ and $\partial_b^T$ and indicate these with a change of subscript to $b$.

Then, by looking at Eq.~\eqref{eq:2dHGP_XZ}, it is easy to check that the following vectors span $\ker H_x$ and $\ker H_z$. We define these $X$ and $Z$ logical representatives as the \textit{canonical logical representatives}:
\begin{itemize}
	\item A canonical basis for the $Z$ logicals is given by:
	\begin{align*}
		\text{First sector: } f^{(i)}_a \otimes \xi^{(j)}_b \text{, where } i\in [k_a^T], j \in [k_b]. \\
		\text{Second sector: } \xi^{(i)}_a \otimes f^{(j)}_b\text{, where } i\in [k_a], j \in [k_b^T].
	\end{align*}
	
	\item A canonical basis for the $X$ logicals is given by
	\begin{align*}
		\text{First sector: } \zeta^{(i)}_a \otimes e^{(j)}_b \text{, where } i \in [k_a^T], j \in [k_b]. \\
		\text{Second sector: }
		e^{(i)}_a \otimes \zeta^{(j)}_b\text{, where } i\in [k_a], j \in [k_b^T].
	\end{align*}
\end{itemize}
Visually, each $X$ and $Z$ canonical logical representatives is either supported on a single column or a single row in one of the two sectors. By counting the number of punctures and dimension of the kernels, we have $k_a^T k_b + k_a k_b^T$ number of $X$ ($Z$)-logical representatives, which is identical to the dimension $k = k_a^T k_b + k_a k_b^T$ of the (co)homology or, the number of logical qubits \cite{Tillich_2014}. 

\begin{figure}
    \centering
    \includegraphics[width=0.5\linewidth]{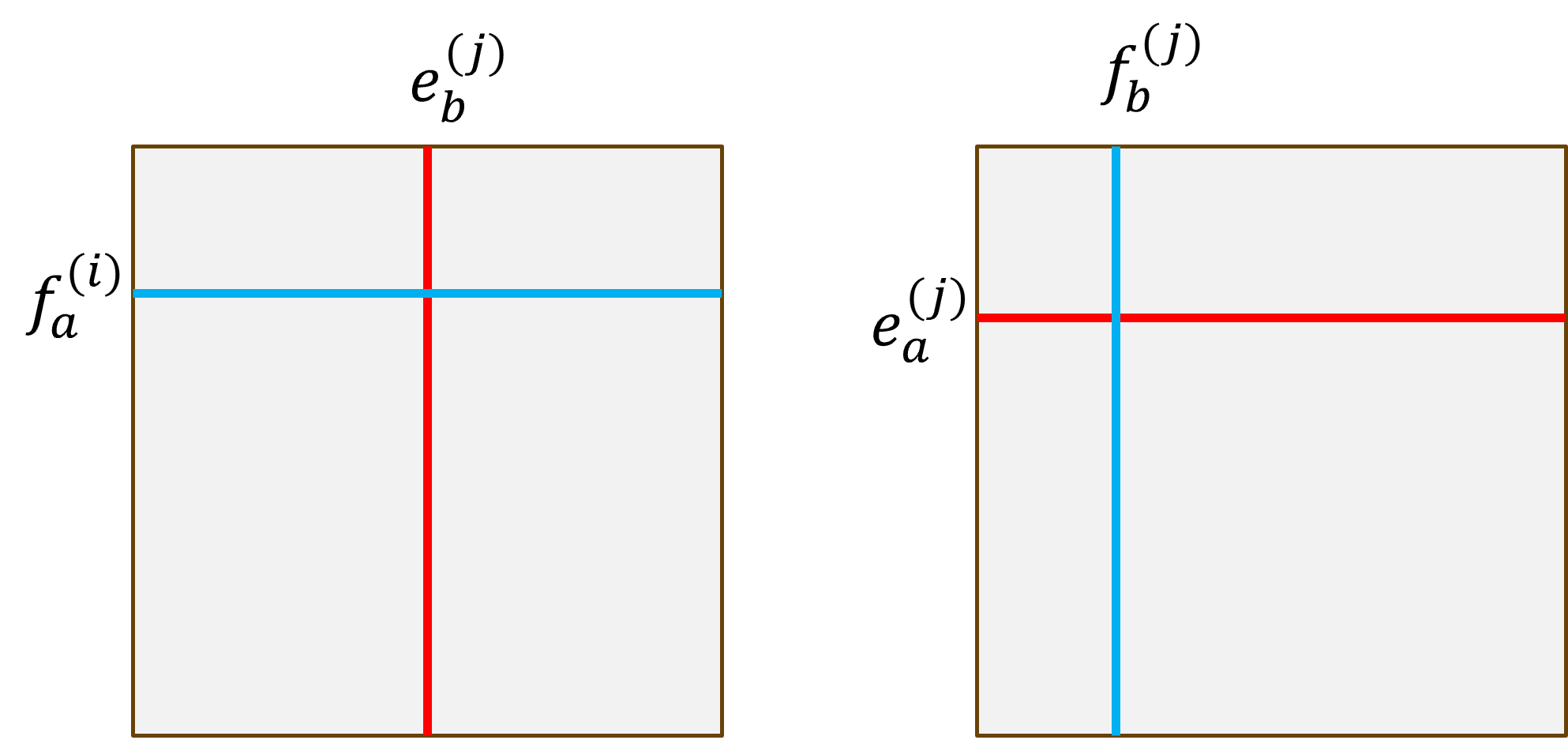}
    \caption{An illustration of $X$ and $Z$ canonical logical representatives supported on either sectors of a 2-dimensional hypergraph product code. An $X$ canonical logical representative, indicated by a red line, is supported on a single column (row) given by the puncture $e_b^{(j)}$ $(e_a^{(i)})$ on the first (second) sector. A $Z$ canonical logical representative, indicated by a blue line, is supported on a single row (column) given by the puncture $f_a^{(i)}$ $ (f_b^{(j)})$ on the first (second) sector.}
    \label{fig:canonical-log-rep}
\end{figure}

From our choice of linearly independent vectors and pivots, we have:
\begin{align}
	\langle f_a^{(i)},f_b^{(j)}\rangle =\delta_{ij}\delta_{ab} = \langle e_a^{(i)},e_b^{(j)}\rangle.
\end{align}
Defining $\xi_b^{(j)},\zeta_b^{(j)}$ by the standard form further indicates that
\begin{align}
	\langle f_a^{(i)},\zeta_b^{(j)}\rangle =\delta_{ij}\delta_{ab} =\langle e_a^{(i)},\xi_b^{(j)}\rangle.
\end{align}
As a result, each $X$ logical representative anti-commutes with only one $Z$ logical representative and vice versa:
\begin{align}\label{Eq: basis 1}
	\text{First sector: } \langle \zeta^{(i)}_a \otimes e^{(j)}_b, f^{(u)}_a \otimes \xi^{(v)}_b \rangle = \delta_{iu}\delta_{jv}, \\
	\text{Second sector: }\langle e^{(i)}_a \otimes \zeta^{(j)}_b ,\xi^{(u)}_a\otimes f^{(v)}_b \rangle = \delta_{iu}\delta_{jv}. \label{Eq: basis 2}
\end{align}
This shows that the canonical logical representatives form a complete basis for the (co)homology classes.

\subsection{General dimensions}\label{sub: generalHGP}

Hypergraph product codes can be defined in high dimensions by taking the tensor products of multiple classical parity check matrices. In the view of algebraic topology, this product is the homological (tensor) product of a set of $1$-chain complexes. 

\begin{definition}
    (Homological product codes) Given a set of $1$-chain complexes over $\mathbb{F}_2$ by $\mathcal{A} = \{ \mathcal{A}^{(1)}, \cdots, \mathcal{A}^{(t)} \}$ where $\mathcal{A}^{(i)}$ is associated with $A^{(i)}: \mathbb{F}^{n_i}_2 \rightarrow \mathbb{F}^{m_i}_2$. The homological product of $\mathcal{A}$ is a $t$-dimensional chain complex over $\mathbb{F}_2$,  denoted $\mathcal{F}_{\bullet}[\mathcal{A}]$,
\begin{align}
    \mathcal{F}_{\bullet}[\mathcal{A}]: \quad \mathcal{F}_t[\mathcal{A}] \xrightarrow{\partial_t} \mathcal{F}_{t-1}[\mathcal{A}]  \xrightarrow{\partial_{t-1}} \cdots \mathcal{F}_1[\mathcal{A}] \xrightarrow{\partial_0} \mathcal{F}_0[\mathcal{A}], 
\end{align}
where the $l$-chain is given by
\begin{align}
    \mathcal{F}_l[\mathcal{A}] = \bigoplus_{\substack{I \subset [t] \\
    |I| = l }}\left( \bigotimes_{i \in I} \mathcal{A}^{(i)}_1 \right) \otimes \left( \bigotimes_{j \notin I} \mathcal{A}^{(j)}_0 \right).
\end{align}
The chain map is given by the graded Leibniz rule. 
Let $x = x_1 \otimes x_2 \otimes \cdots \otimes x_t$ belongs to the sector labeled by $I$ in $\mathcal{F}_l[\mathcal{A}]$. Then
\begin{align}
    \partial_l (x_1 \otimes x_2 \otimes \cdots \otimes x_t)= \sum_{s \in I} x_1 \otimes \cdots \otimes A^{(s)}(x_{s}) \otimes \cdots \otimes x_t, 
\end{align}
where it is defined similarly on other sectors and extends by linearity . 
\end{definition}

A particular feature of the homological product is that the code parameters can be completely characterized by the $1$-chain complexes $\mathcal{A}$, as a consequence of the Künneth formula. 

\begin{lemma}[Künneth formula] Let $\mathcal{A}_{\bullet}$ and $\mathcal{B}_{\bullet}$ be chain complexes over $\mathbb{F}_2$. Let $\mathcal{C}_{\bullet} = \mathcal{A}_{\bullet} \otimes \mathcal{B}_{\bullet}$. Then for $i \in \mathbb{Z}$, 
    \begin{align}
        H_i(\mathcal{C}) \cong \bigoplus_{j \in \mathbb{Z}} H_{j}(\mathcal{A}) \otimes H_{i-j}(\mathcal{B}).
    \end{align}
\end{lemma}

Utilizing the Künneth formula, we can determine the code parameters for the homological product codes consisting of $t$ classical parity checks \cite{Tillich_2014, BH:homological,Zeng_2019,golowich2024quantumldpccodestransversal}.  

\begin{theorem}[Higher homological product code] 
	Given a set of $t$ $1$-chain complexes denoted as $\mathcal{A}$, any three consecutive terms from the homological product $t$-dimensional chain complex $\mathcal{F}_{\bullet}[\mathcal{A}]$ define a quantum CSS code for $1 \leq l \leq t-1$, 
    \begin{align}\label{eq:HGP}
        \mathcal{F}_{l+1}[\mathcal{A}] \xrightarrow{\partial_{l+1}} \mathcal{F}_{l}[\mathcal{A}] \xrightarrow{\partial_l} \mathcal{F}_{l-1}[\mathcal{A}].
    \end{align}
    This code has code parameters $[\![n, k, d]\!]$, where 
    \begin{align}
        \begin{aligned}
            & n = \sum_{[I]} \prod_{i \in I} n_i\prod_{j\in [t]\backslash I} m_j,  \\
	& k = \sum_{I}\prod_{i \in I} k_i \prod_{j \in [t]\backslash I} k_j^T,  \\
	& d_X = \min \left\{\prod_{j \in J}d_j^T : J\subset [t], |J|=r \right\}, \\
	& d_Z=\min \left\{\prod_{i \in I}d_i : I\subset [t], |I|=l\right\},
        \end{aligned}
    \end{align}    
    where $I\subset [t]$ is any subset consisting of $l$ many distinct indices and $J$ is any subset consisting of $r$ distinct indices. 
\end{theorem}

\begin{remark}
    The $Z$ ($X$) distances are determined by examining the Künneth formula on the (co)chain complexes with respect to (co)homologies. For finite dimensional case, the coboundary operators are simply given by the transpose of the boundary operators. 
\end{remark}

We can further observe the close relation of the homological product codes with the underlying classical components by examining a set of basis elements for the logical subspace. 

The logical $Z$ operators' logical basis operators can be written explicitly using the information set associated with $A^{(i)}$ and $A^{(i)T}$:
    \begin{align}\label{eq: z-logical-basis-tdim}
        e^{(a_1)}_{A^{(1)}} \otimes \cdots \otimes e^{(a_r)}_{A^{(r)}} \otimes  \zeta^{(a_{r+1})}_{A^{(r+1)}} \otimes \cdots \otimes \zeta^{(a_t)}_{A^{(t)}}; \quad \cdots \quad ;
        \zeta^{(a_1)}_{A^{(1)}} \otimes \cdots \otimes  \zeta^{(a_l)}_{A^{(l)}} \otimes e^{(a_{l+1})}_{A^{(l+1)}} \cdots \otimes e^{(a_t)}_{A^{(t)}},
    \end{align}
    where $e^{(a_i)}_{A^{(i)}}$ is the unit vector whose nonzero entry $a_i$ associates with an index of an information set $I^T_i \subset [m]$ to $\ker A^{(i)T}$ and $\zeta^{(a_j)}_{A^{(j)}}$ is a logical basis element from $\ker A^{(j)}$. 
    
    The $X$ logical basis operators can be written explicitly as 
    \begin{align} \label{eq: x-logical-basis-tdim}
        \xi^{(a_1)}_{\mathcal{A}^{(1)}} \otimes \cdots \otimes  \xi^{(a_r)}_{\mathcal{A}^{(r)}} \otimes f^{(a_{r+1})}_{\mathcal{A}^{(r+1)}} \otimes \cdots \otimes f^{(a_t)}_{\mathcal{A}^{(t)}}; \quad \cdots \quad ; 
        f^{(a_1)}_{\mathcal{A}^{(1)}} \otimes \cdots \otimes  f^{(a_l)}_{\mathcal{A}^{(l)}} \otimes \xi^{(a_{l+1})}_{\mathcal{A}^{(l+1)}} \otimes  \cdots \otimes \xi^{(a_t)}_{\mathcal{A}^{(t)}},
    \end{align}
    where $f^{(a_i)}_{A^{(i)}}$ is the unit vector whose nonzero entry $a_i$ associates with an index of an information set $I_i \subset [n_i]$ to $\ker A^{(i)}$ and $\xi^{(a_j)}_{A^{(j)}}$ is a logical basis element from $\ker A^{(j)T}$. In this regard, in what follows, we use the notation that $e^{(a_i)}_{A^{(i)}}$ ($f^{(a_i)}_{A^{(i)}}$) as \emph{puncture index} or \emph{information bit} interchangeably.

\subsection{Position of a hyperplane}\label{sub: position}

Pictorially, each canonical logical representative can be supported on a hyperplane of a certain dimension in one of the sectors. Furthermore, multiplying by the same stabilizer across all the qubits with non-trivial support for a canonical logical representative will map this hyperplane to possibly several hyperplanes supported on different sectors, while preserving the dimension of each resulting hyperplane. Therefore, a way of labeling a collection of hyperplanes and specifying their coordinates and dimensions will be useful for keeping track of the support of logical representatives to prove logical gate restrictions for hypergraph product codes. 
	
Let $\{\mathcal{A}^{(i)}\}$ be the set of chain complexes that define the homological chain complex $\mathcal{F}_{\bullet}[\mathcal{A}]$ and define $ \mathcal{F}_{l+1}[\mathcal{A}] \xrightarrow{\partial_{l+1}} \mathcal{F}_{l}[\mathcal{A}] \xrightarrow{\partial_l} \mathcal{F}_{l-1}[\mathcal{A}]$ as a $t$-dimensional hypergraph product code. Recall from Section~\ref{sub: generalHGP}, the $l$-chain $\mathcal{F}_l(\mathcal{A})$ consists of $\binom{t}{l}$ sectors, corresponding to the combinatorial choices of having degree $l$. We identify each of the sector by $S_\mu$ with $\mu = 1, \dots, \binom{t}{l}$. For instance, one of them is of the following form:
\begin{align}
	 \mathcal{A}^{(1)}_{1} \times \cdots \times\mathcal{A}^{(l)}_{1} \times \mathcal{A}^{(l+1)}_{0} \times \cdots \times \mathcal{A}^{(t)}_{0} \cong \mathbb{F}^{n_1}_2 \times \cdots \times \mathbb{F}^{n_l}_2 \times \mathbb{F}^{m_{l+1}}_2 \cdots \times \mathbb{F}^{m_t}_2.
\end{align}
In the following context, the bold letter $\bvec{x}$ would be used to denote general vectors in $\mathcal{F}_l(\mathcal{A})$ including all possible sectors. The restriction of $\bvec{x}$ to one sector $S_\mu$ is denoted by $\mathrm{res}_\mu(\bvec{x}) \in S_\mu$. For vectors defined on a single sector, we simply write $x \in S_\mu$.


\begin{definition}[Coordinates on each sector]\label{def: coordinates on each sector}
Each sector \( S_\mu \subset \mathcal{F}_l(\mathcal{A}) \) corresponds to a unique choice of \( l \) indices out of \( t \), indicating which components appear with degree one. For each such sector \( S_\mu \), we can embed it into a $t$-dim lattice as follows:
\begin{align}
	S_\mu \cong \mathbb{F}_2^{k_\mu} \text{ with } 
	k_\mu = \prod_{i \in I} n_{i} \cdot \prod_{j\in J} m_{j},
\end{align}
where the index sets \( I \), $|I|=l$, and \(J\), $|J|=r$, form a partition of \( \{1, \dots, t\} \), corresponding respectively to the components in degree one and degree zero of the classical codes. Each sector is endowed with a coordinate system with axis labeled by $(x_1,x_2,\cdots,x_t)$, where \( x_{i} \in [n_{i}] \) if $i \in I$ and $x_{j}  \in [m_{j}]$ if $j \in J$, i.e. the ordering of the coordinates follow the same ordering as the indices in $\mathcal{A}^{(i)}$. For simplicity, we would write $t - l = r$ as used at the end of Section~\ref{sub: generalHGP}. 
\end{definition} 


\begin{definition}[Hyperplane]\label{def: hyperplane}
    An $l$-dimensional hyperplane $\xi_i$ is an $l$-dim affine subspace of a $t$-dimensional sector that can be specified by its \textit{orientation} and an $\textit{associated vector}$ (defined below) that together gives the position of the hyperplane on that sector. 
\end{definition}

\begin{definition}[Orientation and associated vector]\label{def: orientation and associated vector}
    We define the orientation of an $l$-dim hyperplane $\xi_i$ in sector $S_\mu$ to be $\mathrm{dir}_\mu (\xi) = b_\mu\subseteq [t]$ where $b_\mu$ is a subset of size $t-l$ with fixed coordinate directions. The complement $[t]\backslash b_\mu$ corresponds to coordinate directions along which $\xi_i$ varies. An associated vector  
    is a $t-l$-dimensional vector that gives the values of the fixed coordinates on the lattice. 
\end{definition}

When there is no ambiguity, the notation $\xi$ would refer to a hyperplane or a vector that is supported on that hyperplane. Following definition~\ref{def: coordinates on each sector}, any $X$ canonical logical representative sector $S_\mu$ is supported on a $r$-dimensional hyperplane with orientation given by $\mathrm{dir}_\mu(\xi_{L_X})= \{j_1, \dots, j_l\}$, and any $Z$ canonical logical representative is supported on a $l$-dimensional hyperplane with orientation $\mathrm{dir}_\mu(\xi_{L_Z})= \{i_1, \dots, i_r\}$. Clearly, $\mathrm{dir}_\mu(\xi_{L_X})$ equals the complement $\overline{\mathrm{dir}_\mu(\xi_{L_Z})}$ of $L_Z$. 

Definition \ref{def:orientationPreserve} from Section~\ref{subsection: artificial scenario} introduces a notion of orientation preserving map, which can be explained as follows: let $V$ be an operator supported on a hyperplane $\xi_V$ on sector $S_{\mu'}$, such that $\mathrm{dir}_{\mu'}(\xi_{L_X})\subseteq\mathrm{dir}(\xi_V)$, where $\mathrm{dir}_{\mu'}(\xi_{L_X})$ is the direction of a canonical $X$ logical representative on $S_{\mu'}$. Then the direction of hyperplanes supporting $\mathrm{res}_{\mu}(UVU^\dagger)$ on any sector $S_\mu$, denoted $\mathrm{dir}_\mu(UVU^\dagger)$, satisfies $\mathrm{dir}_\mu(\xi_{L_X})\subseteq \mathrm{dir}_\mu(UVU^\dagger)$.

\begin{example}
	Given a 2-dimensional homological product code, the first sector has a coordinate system $(x_1,x_2), x_1 \in [n_1],x_2\in [m_2]$. A canonical $X$ logical representative can be supported on a $1$-dim hyperplane $\xi_1$ with orientation given by $\mathrm{dir}_1(\xi_1)=\{2\}$. 
\end{example}

The positions of hyperplanes are useful in finding the intersections of two or more hyperplanes and their resultant dimensions:

\begin{example}
	Given a 3-dimensional homological product code, consider three hyperplanes  specified by $\xi_1$ with $\mathrm{dir}_1{(\xi_1)}= \{2,3\}$ and associated vector ${(2,1)}$ and $\xi_2$ with $\mathrm{dir}_1{(\xi_2)}= \{2\}$ with associated vector $(1)$. The hyperplane from the intersection of the first two planes is given by $u$ with $\mathrm{dir}_1{(u)}= \{2,3\}$ and associated vector $(2,1)$, so this is a 1-dimensional hyperplane.   
\end{example}

Given any vector $v \in S_\mu \cong \mathbb{F}^{n_1}_2 \times \cdots \times \mathbb{F}^{n_l}_2 \times \mathbb{F}^{m_{l+1}}_2 \cdots \times \mathbb{F}^{m_t}_2$, we also wish to know how many hyperplanes that can cover it. 

\begin{definition}
	Given $v \in S_\mu$, recall that the Hamming weight $\vert v \vert$ of $v$ is simply the number of nonzero components of $x$. Its \emph{block Hamming weight} $\vert v \vert_i$ is defined as the number of all nonzero vectors $v(j_1,...,j_{i-1},\text{ - },j_{i+1},...,j_t)$. Geometrically, this is just the total number of lines or 1D hyperplanes with orientation $[t] \setminus \{i\}$ that supports a nontrivial part of $v$. For general hyperplanes with direction $b_\mu \subset [t]$, we denote by $\vert x \vert_{b_\mu}$ the total number of hyperplanes that support the nontrivial part of $x$ with respect to $b_\mu$.
\end{definition}  

To analyze a set of hyperplanes with the same orientation, we use the term hypertube. Informally, one can understand an $r$-dim hypertube as consisting of a thin set of $r$-dim hyperplanes with the same orientation, i.e., with block Hamming weight $> 1$ in some orientation. This notion will also be important for the results in Section~\ref{subsection: constant-depth-general-case}. As a caveat, for any vector $v$, its block Hamming weights with respect to different dimensional hyperplanes and different orientations can be different. For example, the ordinary Hamming weight $\vert x \vert$ record the number of $0$-dimensional hyperplanes that can cover $x$. However, just one $t$-dimensional hyperplane -- the total subspace $S_\mu$ --  cover $x$. It will become clear on how to choose hyperplanes in Section~\ref{sec:constant_depth}.

\subsection{Finding an alternative set of logical representatives} \label{sub: alt}

In Section~\ref{sec: generalized BK}, Theorem~\ref{theorem: generalize bk} requires finding different sets of logical representatives, and roughly speaking, repeatedly taking intersections of the support of various sets of logical representatives. Our goal is to obtain a correctable set at $K_i$ for the smallest possible value of $i$. This is generally extremely difficult. For instance, for the recent explicit constructions of good qLDPC codes \cite{pk22,QuantumTanner2022,DHLV}, logical representatives satisfy well-behaved local checks but it is difficult to understand what the logicals look like globally, and thus it may be challenging to find two sets of representatives whose overlap of supports is correctable. Fortunately, in the setting of hypergraph product codes, we are able to construct nice logical representatives e.g.~using the canonical forms. By adopting  techniques from classical coding theory and applying Bravyi--König-type arguments, we can make use of the different logical representatives to bound the level of Clifford hierarchy for hypergraph product codes. In essence, for hypergraph product codes, it turns out that we can construct different logical representatives with small overlapping support so that we can eventually derive a trivial logical operator supported on a correctable region which then tells us the limitations on the logical gates. 


\begin{lemma}\label{lemma: classical-cleaning}
Let $\mathcal{C}$ be a $[n,k,d]$ classical code with parity check matrix $H$, then for any $\gamma \in [n]$ such that $\gamma \leq (d-1)/{2}$, there exist $h$ in the row span, $\Ima H^T$, such that its restriction $h|_\gamma$ to $\gamma$ is the all-ones vector. 
\end{lemma}
\begin{proof}
	Since $|\gamma|\leq \frac{d-1}{2}$, errors on each bit in $\gamma$ have a unique syndrome. Hence there exist $|\gamma|$ parity checks such that the syndromes restricted to these checks are linearly independent. In other words, we can construct a full-rank parity check matrix when restricted to the bits in $\gamma$ and restricted to these linearly independent checks. Applying elementary row operations, we can obtain a new set of parity check generators such that each parity check detects exactly 1 error on the set $\gamma$, when restricted to $\gamma$.
\end{proof}

The following result will be critical in showing our no-go results for hypergraph product codes in the most general setting.

\begin{lemma}[Information set]\label{lemma: information set}
    Let $\mathcal{C}$ be a classical code with distance $d$. For any index set $T \subseteq [n]$ with $\vert T \vert \geq n - d + 1$, there exists an \emph{information set} $I \subseteq T$ such that the followings hold: 
	\begin{enumerate}
        \item $\vert I \vert = k = \dim \mathcal{C}$. 
        
		\item For any codeword $c \in \mathcal{C}$, if $c \vert_I = \mathbf{0}$, then $c$ is the trivial codeword (i.e. the all-zeros vector on $n$).

		\item The restriction of the parity check matrix $H$ to the complement of $I$ denoted $H|_{\bar{I}}$ is full-rank. Hence, the information set serves as a $k$-puncture set to the parity check matrix $H$. 
	\end{enumerate}
\end{lemma}
    \begin{proof}
	For the first statement, consider the generator matrix $G \in \mathbb{F}^{k \times n}_q$. We claim that any $G|_{T}$ must have full-rank $k$. Suppose not, then there must be a non-zero codeword that lies outside of $G|_{T}$. However, $|\bar{T}| < d$, so it cannot support any nonzero codeword. Hence, by row reduction and column swaps, we can always find $k$ pivots for $G$, so that $G$ is of the form
	\begin{equation*}
		G=\begin{pmatrix}
			\mathbb{I}_k \mid J
		\end{pmatrix}
	\end{equation*}
	for a $k$-by-$k$ identity matrix and some $k$-by-$(n-k) $ matrix $J$. The first $k$ coordinates form an information set $I \subseteq T$.
	
	For the second statement, if $c \vert_I = \mathbf{0}$, then clearly $c$ cannot be generated by $G$ unless $c$ is the all-zero vector.  
	
	Lastly, we show that $H_{\bar{I}}$ must be full-rank, i.e. $\mathrm{rk}(H) = n-k$. This follows from the fact that any vector whose support is entirely contained in \(\bar{I}\) must correspond to a distinct syndrome under \(H\). Suppose, for contradiction, that two such vectors \(u\) and \(v\) satisfy \(Hu = Hv\). Then their difference \(u - v\) lies in the code \(\mathcal{C}\), since \(H(u - v) = 0\). However, \(u - v\) is also supported entirely on \(\bar{I}\), which contradicts the fact that \(I\) is an information set—no nonzero codeword can be zero on all positions in \(I\).  Therefore, all such vectors produce distinct syndromes, implying that the columns of \(H|_{\bar{I}}\) are linearly independent and that \(H|_{\bar{I}}\) has full rank. Hence, any vector in $\operatorname{rs} H$ must have components supported on $H|_{\bar{I}}$, which implies that $I$ forms a $k$-puncture to the parity check matrix $H$. 
\end{proof}

\begin{remark}
    If the parity check matrix $H$ contains linearly dependent rows. That is, the number of rows in $H$ is larger than $n-k$. The above implies that $H|_{\bar{I}}$ is always full-column rank. Denote a puncture set $|\gamma|=t$ such that $\gamma \subset \bar{I}$. Then removing the columns in $\gamma$ would simply reduce the column rank of $H$ by $|\gamma|=t$. Denote the punctured parity check by $H'$ and $\operatorname{rk} H' = n - k -t $ so that the $\operatorname{rk} H' = n - k -t$ while $\ker H' = (n - t) - (n-k-t) = k$. Hence, removing columns purely from $\bar{I}$ would keep the code (kernel) dimension the same. 
\end{remark}

The following corollary of Lemma~\ref{lemma: information set} will be used to find a disjoint set of punctures for hypergraph product codes. 

\begin{corollary}\label{corollary: info set}
	Given any region $R \subset [n]$, suppose $|R| < d$, then one can always find an information set $I$ such that $R \cap I = \emptyset$. 
\end{corollary}

In particular, let $\mathcal{Q}$: $\mathcal{F}_{l+1} \rightarrow \mathcal{F}_l \rightarrow \mathcal{F}_{l-1}$ by the quantum code taken from a $t$-dimensional hypergraph product code. The canonical logical representatives are each supported on a hyperplane in some sector. Restricted to any sector, the support is given by a $r$-dimensional hypertube, i.e., a thin set of $r$-dimensional hyperplanes on $\prod_{i \in I} k_i$ punctures. As long as $k_i < d_i$ for any $i$, one can find a disjoint collection of $X$ logical representatives that are completely disjoint from the original canonical logical representatives in this sector. Similar results hold for the $Z$ logical operators. 
For general logical representatives, not just the canonical ones, Lemma~\ref{lemma: classical-cleaning} is useful in finding new set of logical representatives that has disjoint support from the given one. This technique will be very useful for our results later so we state it here as a lemma.



\begin{lemma}\label{lemma: constant-depth argument}
	Let $d_{\min}$ be the minimum distance of the classical codes used in constructing the hypergraph product code. Consider a canonical logical representative of an $X$- (or $Z$)-type logical operator in one sector, and let $r$ (or $l$) be its dimension. Then, for any collection of \( c \) hyperplanes in the same sector with the same orientation and dimension as the logical representative, where $c \leq \left\lfloor \frac{d_{\min} - 1}{2} \right\rfloor$, there exists an alternative logical representative for the same operator whose support is entirely disjoint from those $c$ hyperplanes. 
\end{lemma}

\begin{proof}
Without loss of generality, assume we are working with a canonical \( X \)-type logical representative. This operator is supported on an \( r \)-dim hyperplane, and the corresponding \( X \)-stabilizer generators are supported on \( (t - r) \)-dim hyperplanes that are perpendicular to it. If the canonical logical representative is disjoint from the $c$ hyperplanes, we are done. Otherwise, suppose it has support on some of the $c$ hyperplanes.

Consider a 1-dimensional hyperplane (i.e., a line) that is perpendicular to the \( r \)-dimensional hyperplane supporting the $X$ logical operator. This line intersects each of the \( l \) given \( r \)-dimensional hyperplanes at one qubit each, so there are a total of \( c \) intersection points.

Now, interpret these \( c \) intersection points as locations of \( Z \)-type errors, and denote this set by $\gamma$. This 1-dimensional line of qubits can be viewed as a classical code: it supports \( X \)-type stabilizers and is responsible for detecting and correcting \( Z \)-errors. Since the underlying classical code has distance at least \( d_{\min} \), and \( c\leq \left\lfloor \frac{d_{\min} - 1}{2} \right\rfloor \), these \( c \) \( Z \)-errors are correctable. Therefore, using Lemma~\ref{lemma: classical-cleaning}, there exists a set of $c$ \( X \)-stabilizer generators such that the syndromes restricted to these generators are linearly independent, and we can obtain a set of \( X \)-stabilizers such that each one anti-commutes with exactly one of the \( Z \)-errors on $\gamma$.

By the translational symmetry of the hypergraph product code, these \( X \)-stabilizers are replicated along the full \( r \)-dimensional hyperplane of the canonical logical operator. Hence, we can multiply the canonical logical representative by appropriate \( X \)-stabilizers---specifically, the stabilizers supported on the same intersection point(s) as the canonical logical representative. The resulting operator intersects trivially with the $c$ hyperplanes. This completes the proof. 
\end{proof}

\section{No transversal non-Clifford gates}\label{sec:transversal}

We now examine the simplest transversal case as defined in Definition~\ref{def: transversal}.
A previous work \cite{burton2020} showed that 2-dimensional hypergraph product codes have transversal gates restricted to the Clifford group under several rather restrictive assumptions on the codes. Their results are only shown to hold for codes that support logical qubits on a single sector. Furthermore, the two seed classical codes are required to have the so-called robustness property which demands that both the generator matrices and parity check matrices of the classical codes have to be simultaneously $k$-bipuncturable. 
Here, we remove all these conditions and establish a universal no-go theorem for transversal non-Clifford gates. We show that any transversal logical gate taking the form of a tensor product of single qubit physical gates belongs to $\P_2$, i.e.,~ the Clifford group, for hypergraph product codes of any dimension. We emphasize that this result applies to all hypergraph product code constructions: it does not impose any constraint on the number of logical qubits involved and does not require any sector restriction or additional conditions on the classical codes. 

As a reminder, implementing multi-controlled-$Z$ gates using several hypergraph product code block to realize multi-controlled logical gates \cite{Chen_2023,Breuckmann2024Cups,golowich2024quantumldpccodestransversal} (see also the examples in Section~\ref{sec: yes-go}) is not prohibited by Theorem~\ref{thm:transversal}. The transversal gate definition being used here does not rule out, for instance, a $\mathrm{C}\mathrm{C}Z$ gate acting on three physical qubits in three separate code blocks. Instead, for our purpose, such $\mathrm{C}\mathrm{C}Z$ gates are classified as gates implemented by constant-depth circuits following Definition~\ref{def: constant-depth}, whose restrictions will be examined later in Section~\ref{sec:constant_depth}.

\begin{theorem}\label{thm:transversal}
Let $\mathcal{Q}$ be a $t$-dimensional ($t\geq 2$) hypergraph product code with $k$ logical qubits for any $k\geq 1$. If $U$ is a logical unitary implementable by transversal gates and the distance of the code satisfies $d\geq 3$, then $U $ is restricted to the Clifford group $ \P_2$.
\end{theorem}

\begin{proof}
Let $\overline{L}_\gamma, \overline{L}_\delta$ be 2 arbitrary logical operators in $\langle \overline{L}_{x_i}, \overline{L}_{z_i}\rangle$, $i\in [k]$ . Without loss of generality, we simply consider $\overline{L}_\gamma = \overline{L}_\delta$, which can be obtained by redefining each of them as the union of $\overline{L}_\delta \cup \overline{L}_\gamma$. Let $L_\delta$ be the canonical logical representatives for $\overline{L}_\delta$ and $L_\gamma$ be some logical representatives for $\overline{L}_\gamma$ to be specified later. 

We can expand $L_\delta$ as the products of $r_1, r_2$ $X$ and $Z$ canonical logical representatives such that each logical representative of the form $L_{p,i}$ is supported by a single puncture, using the canonical logical representatives that we have defined in Section~\ref{sub: 2DHGP} and \ref{sub: generalHGP}. Since there are possibly multiple logical operators with the same puncture, each $L_{p,i}$ is defined to be the product of these logical operators
\begin{align}
	L_\delta = \pm L_{x,1} L_{x,2} \cdots L_{x,r_1} L_{z,1} L_{z,2} \cdots L_{z,r_2}.
\end{align}
Let $U$ be a transversal logical operator. Consider the group commutator:
\begin{align}
	K = L_\gamma ( U L_\delta U^{\dagger} ) L_\gamma^{\dagger} ( U L_\delta^{\dagger} U^{\dagger} ).
\end{align}
We are going to show that $K = \pm 1$ and hence $U$ must be a Clifford gate.

Absorbing the $\pm$ sign into $K$, and inserting identities of the form $\mathbb{I} = U^\dagger U$ and $\mathbb{I} =L_\gamma^\dagger L_\gamma$, we obtain
\begin{align}
	\begin{aligned}
		K = & L_\gamma ( U L_{x,1} \cdots L_{x,r_1} L_{z,1} \cdots L_{z,r_2} U^{\dagger} ) L_\gamma^{\dagger} ( U  L_{z,r_2} \cdots L_{z,1} L_{x,r_1} \cdots L_{x,1}  U^{\dagger} ) \\
		= & ( L_\gamma U L_{x,1} U^{\dagger} L_\gamma^{\dagger} ) (L_\gamma U L_{x,2} U^{\dagger} L_\gamma^{\dagger}) \cdots (L_\gamma U L_{x,r_1} U^{\dagger} L_\gamma^{\dagger} ) \\
		& ( L_\gamma U L_{z,1} U^{\dagger} L_\gamma^{\dagger} ) (L_\gamma U L_{z,2} U^{\dagger} L_\gamma^{\dagger} ) \cdots (L_\gamma U L_{z,r_2} U^{\dagger} L_\gamma^{\dagger} ) \\
		& ( U L_{z,r_2} U^{\dagger} ) ( U L_{z,r_2-1} U^{\dagger} ) \cdots ( U L_{z,1} U^{\dagger} ) \\
		& (U L_{x,r_1} U^{\dagger} ) ( U L_{x,r_1-1} U^{\dagger} ) \cdots (U L_{x,1} U^{\dagger} ). 
	\end{aligned}
\end{align} 

Now consider commuting the last term $U L_{x,1} U^{\dagger}$ to the front. It (anti-)commutes with $U L_{x,i} U^{\dagger}$, $ U L_{z,j} U^{\dagger}$ but the factor $\pm$ is again absorbed into $K$. Consider $[U L_{x,1} U^{\dagger}, L_\gamma U L_{z,i} U^{\dagger} L_\gamma^{\dagger}]$. For each term, since $U$ is transversal, 
\begin{align}
	\mathrm{supp}( U L_{x,i} U^{\dagger} ) = \mathrm{supp}( L_{x,i}), \quad
	\mathrm{supp}(U L_{z,j} U^{\dagger} ) = \mathrm{supp}( L_{z,j}).
\end{align}
Since $L_{x,1}$ is supported on a single hyperplane, and the logical $Z$ representatives are supported either on a different sector, or on the same sector but with overlap of at most a single physical qubit, we find that the support of the commutator is correctable, so the commutator is a trivial logical operator. Thus, we can commute $L_{x,1}$ through with only an extra $\pm$ sign which can again be absorbed into $K$. Now, consider $[U L_{x,1} U^{\dagger}, L_\gamma U L_{x,i} U^{\dagger} L_\gamma^{\dagger} ]$, which depends only on the support of the intersection of $L_{x,1} \cap L_{x,i}$, since $U$ is transversal. If $i \neq 1$, they commute since $\mathrm{supp} (L_{x,1}\cap L_{x,i})= \emptyset$. 

Therefore, we can rewrite $K$ as
\begin{align}\label{eqn: transversal-k-expanded}
	\begin{aligned}
		K = & ( L_\gamma U L_{x,1} U^{\dagger} L_\gamma^{\dagger} U L_{x,1} U^{\dagger} ) \cdots (L_\gamma U L_{x,r_1} U^{\dagger} L_\gamma^{\dagger} U L_{x,r_1} U^{\dagger} )   \\
		& ( L_\gamma U L_{z,1} U^{\dagger} L_\gamma^{\dagger} U L_{z,1} U^{\dagger} ) \cdots (L_\gamma U L_{z,r_2} U^{\dagger} L_\gamma^{\dagger} U L_{z,r_2} U^{\dagger} ) \\
		= & \pm K_1 \cdots K_{r_1} K_{r_1+1} \cdots K_{r_1 +r_2}.
	\end{aligned}
\end{align}
Each $K_i$ is a valid logical operator, so we just need to determine the Clifford hierarchy level for each $K_i$.

Notice that $L_{\gamma}$ appears in each $K_i$. Even though each $L_\gamma$ must represent the same logical operator $\overline{L}_\gamma$, they need not be the same choice of logical representative. We shall proceed to replace $L_\gamma$ with suitable logical representatives to get the best bound for each term $K_i$. 

Without loss of generality, consider the first term $K_1 = ( L_\gamma U L_{x,1} U^{\dagger} L_\gamma^{\dagger} ) (U L_{x,1} U^{\dagger} )$.  First choose $L_\gamma$ to be the canonical logical representatives  
\begin{align}
	L_\gamma = \pm L_{x,1} L_{x,2} \cdots L_{x,r_1} L_{z,1} L_{z,2} \cdots L_{z,r_2}.
\end{align}
However, $L_{x,1}$ can share a large overlapping subset of physical qubits with $U L_{x,1} U^{\dagger}$. But it turns out that we can choose a different representative $L_{x,1}^{(1)}$ instead of $L_{x,1}$, i.e, redefine $L_\gamma$ but keep the rest of the terms the same:
\begin{align}
	L_\gamma =L_{x,1}^{(1)}\cdot L_{x,2}\cdot ...\cdot L_{x,r_1}L_{z,1} L_{z,2} \cdots L_{z,r_2}.
\end{align}
The operator $L_{x,1}^{(1)}$ is chosen by using Lemma~\ref{lemma: constant-depth argument} to find an alternative logical representative that is disjoint from the support of $L_x$. This gives new logical representatives that are no longer supported on the same hyperplane as $U L_{x,1} U^{\dagger}$, that is we have $\mathrm{supp}\,L_{x,1}^{(1)} \cap \supp{L_{x,1}} = \emptyset $.  So
\begin{align}
	\begin{aligned}
		& ( L_\gamma U L_{x,1} U^{\dagger} L_\gamma^{\dagger} ) (U L_{x,1} U^{\dagger} ) \\
		= \pm & ( L_{x,1}^{(1)} L_{x,2} \cdots L_{x,r_1}  L_{z,1} L_{z,2} \cdots L_{z,r_2} ) U L_{x,1} U^{\dagger} \\
		& (L_{x,1}^{(1)} L_{x,2} \cdots L_{x,r_1}  L_{z,1} L_{z,2} \cdots L_{z,r_2} )^{\dagger}  (U L_{x,1} U^{\dagger} ) \\
		= \pm & (  L_{z,1} L_{z,2} \cdots L_{z,r_2}  L_{x,1}^{(1)} L_{x,2} \cdots L_{x,r_1} ) U L_{x,1} U^{\dagger} \\
		& ( L_{z,1} L_{z,2} \cdots L_{z,r_2}  L_{x,1}^{(1)} L_{x,2} \cdots L_{x,r_1} )^{\dagger}  (U L_{x,1} U^{\dagger} ) \\
		=  \pm & (  L_{z,1} L_{z,2} \cdots L_{z,r_2} ) U L_{x,1} U^{\dagger} 
		( L_{z,1} L_{z,2} \cdots L_{z,r_2} )^{\dagger}  (U L_{x,1} U^{\dagger} ).
	\end{aligned}
\end{align}
Lastly, using again the fact that the support of $L_{z,i}$ overlaps with the support of $U L_{x,1} U^{\dagger}$ at at most a single physical qubit, we conclude that the above operator must be a trivial logical operator. 

Therefore, we have shown that each $K_i$ in Eq.~(\ref{eqn: transversal-k-expanded}) is an identity logical operator  modulo at most a $\pm 1$ phase. Since $K$ is in $\P_0$, and this holds true for any sequence of $\{\overline{P}_j\}_{j \in [2]}$, using the generalized Bravyi--K\"{o}nig theorem we formulated as Lemma~\ref{theorem: generalize bk}, we conclude that $U \in \P_2$, i.e., $U$ is in the Clifford group. 
\end{proof}

\begin{remark}
	With the fact that Clifford operators form the group $\mathcal{P}_2$, by Lemma~\ref{lemma: clifford-hierachy-basis}, a simpler proof of Theorem~\ref{thm:transversal} can be made as follows.
	
	Let $\overline{L}_\delta, \overline{L}_\gamma$ be arbitrary elements from the generating set of logical operators $\{\overline{L}_{x_i}, \overline{L}_{z_i}\}$ for $i \in [k]$, and let their corresponding canonical logical representatives be $L_\delta$ and $L_\gamma$.
	
	Define:
	\[
	K = L_\gamma ( U L_\delta U^{\dagger} ) L_\gamma^{\dagger} ( U L_\delta^{\dagger} U^{\dagger} ).
	\]
	
	We consider three possible cases for the pair $(L_\delta, L_\gamma)$, based on the sets:
	\begin{itemize}
		\item $A = \{\{L_{x_i}, L_{x_j}\}, \{L_{z_i}, L_{z_j}\}\}$ with $i \ne j$,
		\item $B = \{\{L_{x_i}, L_{x_i}\}, \{L_{z_i}, L_{z_i}\}\}$, i.e., $L_\delta = L_\gamma$,
		\item $C = \{\{L_{x_i}, L_{z_j}\}\}$, where one is a logical $X$ and the other a logical $Z$.
	\end{itemize}

	In cases $A$ and $C$, the support of the intersection of $L_\delta$ and $L_\gamma$ is at most a single physical qubit, so it follows that $K = \pm 1$, by Lemma~\ref{lemma: phase for trivial operator}. In fact, for the case of $A$, the support of $L_\delta$ and $L_\gamma$ is disjoint.
	
	In the case for $B$, we can apply Lemma~\ref{lemma: constant-depth argument} to construct an alternative logical representative $L_\gamma'$ such that it is disjoint from the canonical logical representative $L_\gamma$. Redefining $L_\gamma$ to be $L_\gamma'$, we find that the support of $K$ becomes empty, so $K = 1$.
	
	In all three cases, we conclude that $K = \pm 1$. Then by Lemma~\ref{lemma: clifford-hierachy-basis}, it follows that $U \in \mathcal{P}_2$.
\end{remark}

Theorem~\ref{thm:transversal} can be immediately generalized to multi-HGP code blocks where we take arbitrary tensor products of stabilizers and logicals from each block. Consequently, the theorem indicates that there is no hope of directly getting transversal $T$ gate on any hypergraph product codes.
{
\begin{example} 
	It is instructive to compare this with the transversal implementation of $T$ gates on color codes \cite{Bombin_2007,Bombin_2013,Kubica2015}. Let $M$ be a $t$-dimensional manifold with boundary. Suppose it can be triangulated into a simplicial complex, still denoted by $M$.  Let $\mathtt{TC}_k(M)$ denote the toric code defined on the complex with qubits on $k$-simplices, $X$-stabilizer generators and $Z$-stabilizer generators supported on $(k-1)$- and $(k+1)$-simplices respectively. If the complex $M$ is $(t+1)$-colorable, we can define the color code $\mathtt{CC}_k(M)$ with qubits on $t$-simplices, $X$-stabilizer generators and $Z$-stabilizer generators supported on $(t - k - 2)$- and $k$-simplices respectively. To guarantee the commutativity of the stabilizer generators, they need to be carefully defined near the boundary (see, e.g., \cite{Kubica2015Unfolding} for more details). Then the so-called \emph{color code unfolding} \cite{Yoshida_2011,Bomb_2014,Kubica2015Unfolding} shows that there exists a local Clifford unitary $U$, cell complexes $M_i$ obtained from $M_i$, stabilizer groups $S_1,S_2$ of certain ancilla qubits such that
	\begin{align}\label{eq:unfolding}
		\mathtt{CC}_k(M) \otimes S_1 =  U [\mathtt{TC}_{k+1} (\#_i M_i) \otimes S_2 ] U^\dagger,
	\end{align}
	where $\#_i M_i$ means attaching $M_i$ together by identifying some of their boundaries. The ancilla qubits can be added locally on the side of either the color code or the toric code, depending on the local structure of the complexes.
	
	With or without boundary, by Theorem~\ref{thm:transversal}, toric codes do not support transversal $T$ gates and their higher order analogues $R_l = \text{diag}(1, e^{2\pi i /2^l})$. On the other hand, it is well known that these gates can be transversally implemented on $l$-dimensional color codes by physical $R_l$ and $R_l^\dagger$. Eq.~\eqref{eq:unfolding} reveals the relationship between two cases: the implementation of these logical non-Clifford gates on toric codes is enabled by introducing a constant-depth local Clifford circuit $U$ and ancilla qubits, which is consistent with Theorem~\ref{thm:transversal}.
\end{example} 
}

\section{Constant-depth circuit gates}\label{sec:constant_depth}

In Section~\ref{sec:transversal} we have shown that strictly transversal gates alone are not capable of implementing any non-Clifford logical gates in hypergraph product codes. However, as explained earlier, relaxing the requirement of transversal gates to constant-depth circuits still guarantees fault tolerance, and by the above result they are necessary to access non-Clifford gates.

As hypergraph product codes may not be geometrically local, with stabilizers having support on all sectors of a given code, we run into the problem of how the constant-depth circuit change the support of a given logical representative. 

Rather than analyzing arbitrary constant-depth circuits which would be intractable, we focus on a natural constrained class: circuits that map hyperplanes to hyperplanes or hypertubes to hypertubes. The allowable hyperplanes is as defined in Section~\ref{sub: position}, those that can be specified by an orientation and an associated vector. More precisely, the assumption is given below:

\begin{assumption}\label{assmuption0}
(Dimension-preserving) Suppose $A$ is an operator with support on a $l$-dimensional hypertube in any sector. Under a constant-depth circuit that implements the unitary gate $U$, the support of $UAU^\dagger$ is given by a constant number of hypertubes supported on (possibly multiple) sector(s). Further, the dimension of each of them must be at most $l$.
\end{assumption}

Note that this assumption neither requires the code to be geometrically local nor the gate to be short range (cf. the original Bravyi--König theorems in \cite{bravyi_classification_2013}), since the hypertubes for $UAU^\dagger$ can be supported across different sectors. Even in one sector, hyperplanes need not be stacked together.

This assumption ensures the dimension of hyperplanes does not increase under the constant-depth circuit, so that we can eventually argue for correctability as we recursively attempt to perform dimension reduction through finding alternative logical representatives and taking intersections. It is also motivated by the way logical gates are implemented to perform a multi-qubit controlled gate in the high dimensional toric codes.

With Assumption \ref{assmuption0}, our analysis proceeds in three parts: (1) orientation-preserving (2) sector-preserving (3) shaping-preserving in the following subsections.
\begin{enumerate}
    \item In Section~\ref{subsection: artificial scenario}, we study the artificial scenario where the constant-depth circuit is orientation-preserving, namely it maps hyperplanes with a given orientation to hyperplanes with the same orientation (see Definition~\ref{def:orientationPreserve}). Theorem~\ref{thm:orientation} demonstrates that without different orientations, non-Clifford gates cannot be realized. We use this artificially constructed scenario to argue that it is important to build constant-depth circuits that change orientations of hyperplanes for non-Clifford gates. 
    
    \item In Section~\ref{subsection: single sector constant depth}, we consider the case where logical gates are confined to a single sector, and show that for any family of hypergraph product codes with growing distance, the logical gates are restricted to $\mathcal{P}_{\lfloor r/l\rfloor +1}$. This is a more restricted model leading up to Section~\ref{subsection: constant-depth-general-case}. where instead of single sector we consider multiple sector. However, unlike the previous case in Section~\ref{subsection: artificial scenario}, we now consider a constant-depth circuit that can, roughly speaking, send a hyperplane to hyperplanes with different orientations.

   \item In Section~\ref{subsection: constant-depth-general-case}, we address the most general setting: constant-depth circuits entangle qubits across multiple sectors and map hyperplanes to hyperplanes with potentially different orientations. The requirement on the distance of the code depends on whether we impose the shape-preserving condition (see Definition \ref{def: shape-preserving}) for the constant-depth circuit. In general, the distance requirement is less if the circuit is also shape-preserving. 
\end{enumerate}

Collectively, these results and Theorem~\ref{thm:transversal} are summarized in Table~\ref{tab:summary of results}.

\begin{table}
\begin{tabular}{|p{3.8cm}|p{3cm}|p{3.8cm}|p{3.8cm}|}
\hline
 & Transversal
 & \multicolumn{2}{c|}{Constant-depth} \\
 \hline
 &
 & Orientation preserving 
 & Single/Multi-sector constant-depth circuit \\
\hline
Arbitrary $k$ 
& $\mathcal{P}_2$ 
& $\mathcal{P}_2$ 
& Not known \\
\hline
$c(t,r) k^{g(t,r)} < d$ 
& $\mathcal{P}_2$ 
& $\mathcal{P}_2$ 
& $\mathcal{P}_{\lfloor r/l\rfloor +1}$ \\
\hline
\end{tabular}
\caption{Summary of fundamental limitations on logical gates for all dimensions of hypergraph product codes. Here, $l$ ($r$) is the dimension of a canonical $Z$ ($X$) logical representative, with $l+r =t$ the dimension of the hypergraph product code. We also set $k$ to be the code dimension, and $d = \min\{d_i,d_i^T\}$ be the minimum classical code distance. The functions $c(t,r)$ and $g(t,r)$ are determined based on more refined assumptions on the circuits presented in Section~\ref{subsection: single sector constant depth} and \ref{subsection: constant-depth-general-case}. The table shows the limitation across different scenarios. The \textit{rows} represent constraints on the number of logical qubits. The \textit{columns} compare transversal and constant-depth implementations under various structural assumptions. Transversal gates in HGP cannot implement non-Clifford gates. For constant-depth circuits, if the unitary operators are orientation preserving then it is at most a Clifford gate, but in general it can implement a logical gate in $(\lfloor r/l\rfloor +1)$-th level of Clifford hierarchy, if the distance of the code is sufficiently larged compared to the number of logical qubits.}
\label{tab:summary of results}
\end{table}

\subsection{Orientation-preserving constant-depth circuits}\label{subsection: artificial scenario}

To show that having hyperplanes of operators with different orientations is important to obtain non-Clifford gates, we construct an artificial case of orientation-preserving constant-depth circuit.

\begin{definition}\label{def:orientationPreserve}
	Let $V$ by arbitrary logical operator supported on any hypertube $\bvec{\xi}$ in any sector. A constant-depth circuit $U$ is \emph{orientation-preserving} if it satisfies Assumption \ref{assmuption0} and $U P U ^\dagger$ consists of hypertubes with the same orientation possibly across different sectors.
\end{definition}

Explicitly, if $V$ is supported on a single $r$-dim hyperplane that has the same orientation as an $X$ (or $Z$) canonical logical representative in the same sector as $V$, then the support of $UVU^\dagger$ consists of hyperplanes with dimension at most $r$. These hyperplanes can be supported on multiple sectors, but each hyperplane has the same orientation with respect to either $X$ or $Z$ canonical logical representatives in the same sector. In the following proof, it is sufficient to work with individual logical basis elements. As a result, we only consider hyperplanes instead of hypertubes. Since $U$ is a constant depth circuit, the single hyperplane is mapped to at most $c$ hyperplanes.

\begin{theorem}\label{thm:orientation}
	Suppose a unitary operator $U$ is orientation-preserving and is implementable by a constant-depth circuit that maps a hyperplane to at most $c$ hyperplanes on a \( t \)-dimensional hypergraph product code where the canonical logical \( X \) and \( Z \) representatives are supported on hyperplanes of dimensions \( r\) and \( l \) respectively, with \( l + r = t \). Then, if $d>c^2 +c =O(1)$, \( U \) implements an encoded gate from the Clifford group $\mathcal{P}_2$.
\end{theorem}
\begin{figure}[ht]
    \centering
    \includegraphics[width=0.75\linewidth]{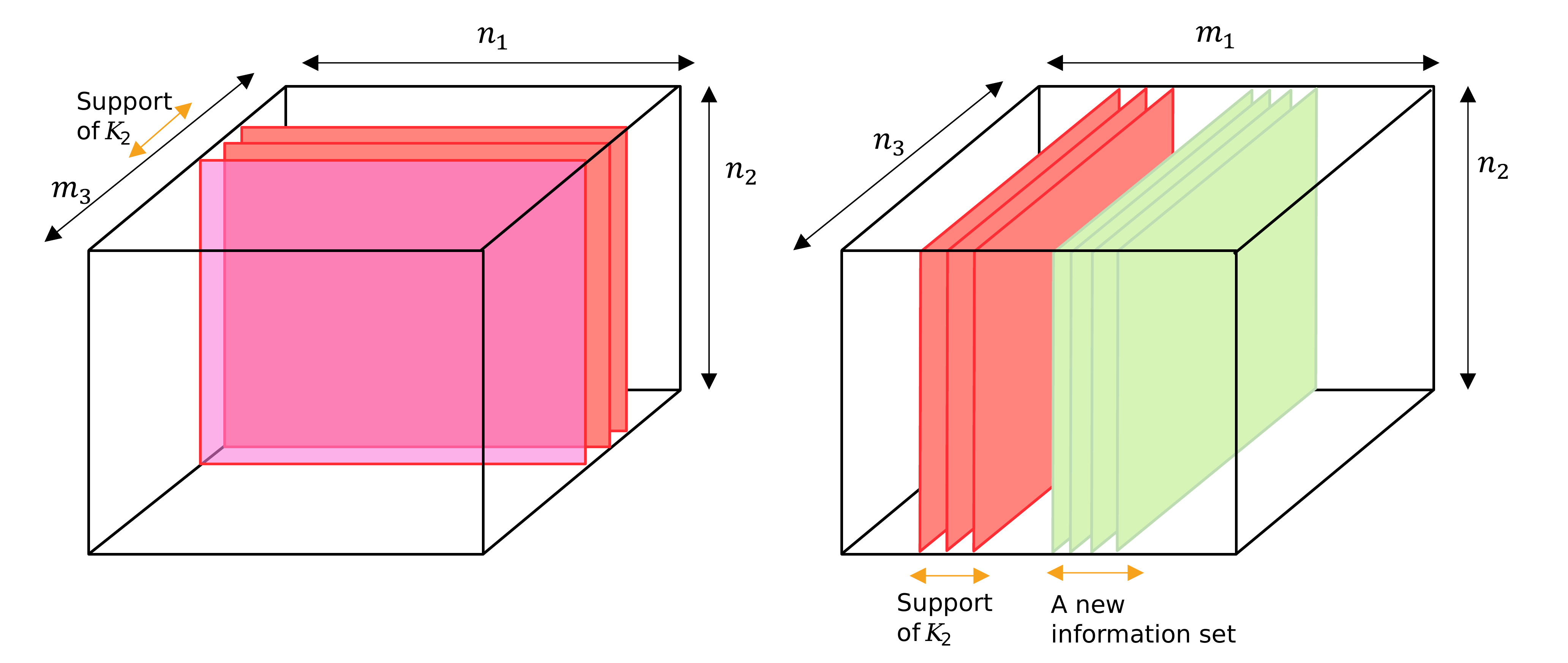}
    \caption{\textbf{Support of logical operators in two sectors of a hypergraph product code under a constant-depth circuit.} This diagram illustrates the case where both \( L_1 \) and \( L_2 \) are \( X \)-type logical operators, with the canonical representatives for $L_1$ and $L_2$ found on the left and right sector respectively, but the argument can be easily extended from here. 
    We begin with a pink hyperplane representing the support of a canonical \( X \)-type logical operator \( L_1 \). 
Under conjugation by a constant-depth circuit \( U \), the resulting operator \( K_1 = U L_1 U^\dagger \) has support on a bounded number of red hyperplanes in each sector. 
In the right sector, a new logical basis \( \mathcal{B}_2 \) is chosen such that each logical operator \( L_2 \in \mathcal{B}_2 \) is supported on a single green hyperplane that is disjoint from the red hyperplanes in that sector. 
As a result, the intersection \( \mathrm{supp}(L_2) \cap \mathrm{supp}(K_1) \) is empty, so the commutator \( K_2 = K_1 L_2 K_1^\dagger L_2^\dagger \) acts trivially on the code space. 
}

    \label{fig:orientation-preserving-map}
\end{figure}
\begin{proof}
	Let $\mathcal{Q}$ be a $t$-dim hypergraph product code with $k$ logical qubits. WLOG, assume, $r \geq l$. Assume that $U$ maps a hyperplane to at most $c = O(1)$ hyperplanes. Assume that the code has a distance $d >c^2 $. 

    Following the proof in Lemma~\ref{lemma: clifford-hierachy-basis}, we will show that for any $L_1,L_2$ in some basis set $\mathcal{B}_{i = 1,2}$, we can find logical representatives for $L_1$ and $L_2$, such that for the following equations
    \begin{align*}
        K_1 = U L_1 U^\dagger,\\
        K_2 = K_1 L_2 K_1^\dagger L_2^\dagger,
    \end{align*}
    we have that $K_2$ is supported on a correctable set, so this implies that $K_2 = \pm 1$, then $U \in \P_2$, the Clifford group. 
  
	First let $\overline{L_1}$ be any logical operator in the canonical basis set $\mathcal{B}_1 = \{ L_{x_i},L_{z_i}\}, i \in [k]$, i.e. $L_1$ is either a single qubit $X$ or $Z$ logical operator. Let $L_1$ be the canonical $X$ (or $Z$) logical representative for $\overline{L_1}$. The support of $K_1$ is thus given by at most $c$ hyperplanes of at most $r$ (or $l$) dimension. Since $U$ is orientation preserving, the orientation of hyperplanes on each sector $S_\mu$ satisfies $\mathrm{dir}_\mu(L^{(\mu)})\subseteq \mathrm{dir}(\mathrm{res}_\mu (K_1))$, where $L^{(\mu)}$ is a canonical logical $X$ (or $Z$) representative on $S_\mu$. See Figure~\ref{fig:orientation-preserving-map} of an illustration of an orientation preserving constant depth circuit $U$ on an $X$ logical representative. 
	
	We will examine the commutator $K_2 = K_1 L_2 K_1^\dagger L_2^\dagger$, with $L_2$ specified later.
    
	Since $K_1$ consists of at most $c$ hyperplanes, the number of hyperplanes with the same orientation as an $X$ (or $Z$) canonical logical representatives in any sector is at most $c$. If $d>c$, then we can use Lemma~\ref{lemma: information set} and Corollary \ref{corollary: info set} on each sector, so that we have a set of $X$ (or $Z$) logical representatives that are disjoint from the hyperplanes of $K_1$ and then we will pick the canonical $Z$ (or $X$) logical representatives. 
	
	For the new set of $X$ (or $Z$) logical representatives, there exist a choice of basis $\mathcal{B}_2$ of logical operators such that each element has a corresponding logical representative supported on a single hyperplane. This basis is not necessarily the same as the canonical one. Pick any $\overline{L}_2 \in \mathcal{B}_2$ and denote the logical representative as $L_2$.
	
	Applying Lemma~\ref{lemma: constant depth support}, the support of $K_2$ is given by:
	\begin{align*}
		\supp[K_1 V K_1^\dagger] \cup \supp[V],
	\end{align*} 
	where $V = L_2|_{\supp[K_1]}$, so that $\supp[V] = \supp[L_2] \cap \supp[K_1]$.
	
	The support of $V$ consists of at most $c$ physical qubits, because it is given by the intersection of a single $r$ (or $l$) dimensional hyperplane in one of the sectors with the $c$ hyperplanes from $K_1$.
	
	So the support of $K_2$ consists of at most $c^2 + c= O(1)$ number of physical qubits which forms a correctable subset, assuming that $c^2 +c <d$. 
	
	Since $K_2$ is a trivial logical operator for any $\overline{L}_2 \in \mathcal{B}_2$, $K_2$ must be a trivial logical operator for any logical Pauli operator in $L_1$. This implies that $K_1$ must be a Pauli operator itself. 
	
	Further $K_1 = UL_1U^\dagger$ for any $L_1$ that acts non-trivially on a single logical qubit. Since $\P_1$ forms a group, this is true if $L_1$ is any element in $\P_1$. Thus, $U$ must be in the Clifford group $\P_2$.
\end{proof}

\subsection{Single-sector constant-depth circuits}\label{subsection: single sector constant depth}

In this subsection, we study the limitations of logical gates implemented by constant-depth circuits that act within a single sector of a hypergraph product code. Building on the previous subsection, we know that orientation-preserving circuits are too restrictive to allow nontrivial gates. Here, we consider the logical gate $U$ implementable by a constant-depth circuit that satisfies Assumption \ref{assmuption0}. Unlike before, we do not require the circuit to preserve hyperplane orientations; instead, we only assume that it maps operators to other operators supported within the same sector. Our goal is to understand the constraints that this locality imposes on fault-tolerant logical gates acting on an arbitrary subset of the $k$ logical qubits.

We will adopt the previous notations in Section~\ref{sub: position} for the rest of this subsection, and apply them to a single sector of an arbitrary hypergraph product code. We describe the notations below again for convenience of reference. 

\paragraph{Set-up} Let $\mathcal{Q}$ be the $t$-dim hypergraph product code defined using the homological product of $t$ classical codes, $A^{(i)}: \mathbb{F}^{n_i}_2 \rightarrow \mathbb{F}^{m_i}_2$, encoding a total of $k$ logical qubits. Let $U$ be a logical gate implementable by a constant depth circuit that maps a hyperplane to at most $c$ hyperplanes. Let $\{\overline{L}_j\}_{j\in[\lfloor\frac{r}{l}+1\rfloor]}$ be $r/l+1$ arbitrary logical operators in $\langle \overline{L_{x_i}},\overline{L_{z_i}}\rangle$, $i\in [k]$ and let $L_i$ be the logical representative for $\overline{L}_i$ to be determined later. Let $L_i = L_{i,X}\cdot L_{i,Z}$, where $L_{i,X}(L_{i,Z})$ are the $X(Z)$ logicals of $L_i$. Let $\kappa_i=\kappa_{i,x}+\kappa_{i,z}$ denote the number of non-trivial logical operators in defining $L_i$. Let the hyperplanes for canonical $X$ and $Z$ logical representatives be $r$ and $l$-dimensional respectively, such that $r + l = t$.

$U$ is assumed to implement a logical gate on a single sector $S_{\mu =1}$, with the partitioning given by index sets $\{1,2,\cdots,r \}$ and $\{r+1,\cdots, t\}$ for components with degree one and zero respectively. WLOG let the lattice's coordinate system be $(x_1, \cdots, x_t)$, with
\begin{align*}
    &x_{i} \in [n_i], \quad 1\leq i\leq l,\\
    &x_{j} \in [m_j] , \quad l+1\leq j\leq t.
\end{align*}

Then, a canonical $X$ logical representative $L_X$ on $S_1$ has orientation given by $\mathrm{dir}(L_X) = (r+1,\cdots,t)$, while a canonical $Z$ logical representative $L_Z$ has orientation given by $\mathrm{dir}(L_Z) = (1,\cdots, r)$ (Definition~\ref{def: hyperplane}). We will also use Definition~\ref{def: orientation and associated vector} heavily for the proof below.  

\begin{theorem}\label{thm:single-sector}
Suppose a unitary operator \( U \) is implementable by a constant-depth circuit that maps logical operators supported on hyperplanes within a single sector of a \( t \)-dimensional hypergraph product code to operators supported on hyperplanes (possibly with different orientations) within the same sector. Let the canonical logical \( X \) and \( Z \) representatives be supported on hyperplanes of dimensions \( r \) and \( l \) respectively, with \( l + r = t \). Then, if the distance is at least $d=O(\kappa_{\max\{[\lfloor\frac{r}{l}\rfloor+1]\}}^{2^{\lfloor\frac{r}{l}\rfloor-1}})$, \( U \) implements an encoded gate from the \( \mathcal{P}_{\lfloor r/l \rfloor + 1} \) level of the Clifford hierarchy.
\end{theorem}

\begin{proof}
Assume \( r > l\); the case \( l = r \) will be treated at the end.

There are at most \( k_X = \prod_{i = l+1}^t k_i^T \) hyperplanes of dimension \( r \), and \( k_Z = \prod_{j = 1}^l k_j \) hyperplanes of dimension \( l \), corresponding to the canonical logical representatives of the logical operators. Note that multiple (possibly non-constant) logical qubits may be supported on the same hyperplane.

Let \( K_1 = UL_1 U^\dagger \). Let \( \{\xi_i\}_i \) and \( \{\mu_j\}_j \) denote the sets of at most \( r \)-dimensional and \( l \)-dimensional hyperplanes that contain the support of \( UL_X U^\dagger \) and \( UL_Z U^\dagger \), respectively. Suppose \( |\{\xi_i\}| = c \kappa_{1,x} \) and \( |\{\mu_j\}| = c \kappa_{1,z} \) with \( \kappa_1 = \kappa_{1,x} + \kappa_{1,z} \) being defined in the set-up above. We assume \( \kappa_1 = O(1) \) is a constant. (Again emphasis that it is possible that a non-constant number of logical operators may be supported on each of the \( \kappa_1 \) hyperplanes.) Also note that these hyperplanes may have non-trivial intersections.

We study the support of $K_2 = K_1 L_2 K_1^\dagger L_2^\dagger$, by first carefully choosing an appropriate logical representative for $L_2$. For each hyperplane $\gamma\in \{\xi_i\}\cup \{\mu_j\}$, we first faction them into one of the following 4 possibilities:
(Recall that $\mathrm{dir}(\cdot)$ indicates the orientation of a hyperplane as in Definition~\ref{def: orientation and associated vector}).
\begin{align}
	& \mathrm{Case \; 1:\;}\mathrm{dir}(\gamma) \cap \mathrm{dir}(L_X) \neq \emptyset, \quad \mathrm{Case \; 2:\;}
	\mathrm{dir}(\gamma) \cap \mathrm{dir}(L_Z) \neq \emptyset, \\
	& \mathrm{Case \; 3:\;}\mathrm{dir}(\gamma) \cap \mathrm{dir}(L_X) = \emptyset, \quad
	\mathrm{Case \; 4:\;}\mathrm{dir}(\gamma) \cap \mathrm{dir}(L_Z) = \emptyset,
\end{align}

\begin{table}[h]\label{fig: 4 cases of intersections}
\centering
\renewcommand{\arraystretch}{1.5}
\small
\begin{tabular}{|p{2.2cm}|p{3cm}|p{3cm}|p{3cm}|p{3cm}|}
\hline
Hyperplanes $\gamma$
& \textbf{Case 1:} 
& \textbf{Case 2:} 
& \textbf{Case 3:} 
& \textbf{Case 4:} \\
&$\mathrm{dir}(\gamma) \cap \mathrm{dir}(L_X) \neq \emptyset$ 
& $\mathrm{dir}(\gamma) \cap \mathrm{dir}(L_Z) \neq \emptyset$
&$\mathrm{dir}(\gamma) \cap \mathrm{dir}(L_X) = \emptyset$
&$\mathrm{dir}(\gamma) \cap \mathrm{dir}(L_Z) = \emptyset$\\
\hline
$\gamma \in \{\xi_i\}_i$ 
& \begin{tabular}{c}
$\checkmark$ possible \\
max dim: $\leq r$
\end{tabular}
& \begin{tabular}{c}
$\checkmark$ possible \\
max dim: $\leq l$
\end{tabular}
& \begin{tabular}{c}
$\checkmark$ possible \\
max dim: $\leq r - l$
\end{tabular}
& \begin{tabular}{c}
$\checkmark$ possible \\
max dim: $0$
\end{tabular} \\
\hline
$\gamma \in \{\mu_j\}_j$ 
& \begin{tabular}{c}
$\checkmark$ possible \\
max dim: $\leq l$
\end{tabular}
& \begin{tabular}{c}
$\checkmark$ always \\
max dim: $\leq l$
\end{tabular}
& \begin{tabular}{c}
$\checkmark$ possible \\
max dim: $0$
\end{tabular}
& \begin{tabular}{c}
$\times$ not possible\\
\end{tabular}\\
\hline
\end{tabular}

\caption{Summary of splitting hyperplanes supporting $K_1$ into the 4 cases. 
Summary of factioning hyperplanes into 4 cases depending on the direction of each hyperplane $\gamma$ (from $\{\xi_i\}_i$ and $\{\mu_j\}_j$) and the directions of canonical logical operators $L_X$, $L_Z$. Only Case 4 with $\gamma \in \{\mu_j\}_j$ is not possible because $r > l$. The maximum dimension of intersection after minimizing over the possible logical representatives of $L_X$ or $L_Z$ is indicated in each case.}
\end{table}

We explain how to construct the logical representatives by the method of cleaning from each coordinate, so that for all hyperplanes $\gamma\in\{\xi_i\}\cup \{\mu_j\}$ they intersect trivially with the new logical representatives for $L_2$ if they fall into case 1 and 2, leaving $K_2$ with support on at most some $(r-l)$-dim hyperplanes and $0$-dim hyperplanes (see Table~\ref{fig: 4 cases of intersections}).

Let $V_{1} = \{\gamma \in \{\xi_i\}\cup \{\mu_j\}; \mathrm{dir}(\gamma) \cap \{r+1,\cdots,t\}\neq \emptyset \}$ (Hyperplanes in Case 1). Then we can always pick a set of at most $c\kappa_1$ (coordinate,value) pairs on the last $l$ coordinates such that for each $\gamma \in V_1$, there exist at least one pair that coincides with the value of the $\gamma$ at the same coordinate. Denote this set by $\mathrm{s}_1$, $|\mathrm{s}_1|\leq c\kappa_1$. Assume that $c\kappa_1 \leq \frac{d-1}{2}$, where $d = \min \{d_i, d_j^T\}, 1\leq i\leq r, r+1\leq j\leq t $. To find a second set of $X$ logical representatives for $\overline{L}_2$, we start with its canonical logical representative and apply Lemma~\ref{lemma: classical-cleaning} to the last $l$ coordinates, so that the $X$ logical representative is cleaned away from the coordinates in $\mathrm{s}_1$. After the cleaning, the set of hyperplanes supporting the new $X$ logical representative, denoted $\overline{L_{2,X}}$ by $\{\xi_{X,j}^{(2)}\}_j$, is disjoint from the support of $\gamma \in V_1$, since each new hyperplane differs from each $\gamma$ by at least one coordinate in their positions. Assuming qLDPC code, with weight of stabilizer generators at most $w$, there are at most $c\kappa_1\kappa_{2,X}w+\kappa_{2,X}$ $r$-dim hyperplanes with $\mathrm{dir}(L_X)$. 

The same analysis as above applies for finding the second set of $Z$ logical representatives. Since Case 4 is not possible, Case 2 contains the entire set of $\{\mu_j\}_j$. We again require $c\kappa_1 \leq \frac{d-1}{2}$, where $d = \min \{d_i, d_j^T\}, 1\leq i\leq r, r+1\leq j\leq t $. $L_{2,Z}$ is supported on $\{\xi_{Z,j}^{(2)}\}_j$
and there are at most $c\kappa_1\kappa_{2,Z}w+\kappa_{2,Z}$ $l$-dim hyperplanes with $\mathrm{dir}(L_Z)$. 

Next, we show that \( K_2 \) is supported on hyperplanes of dimension at most \( r-l\), using the following lemma. 

\begin{lemma}
    \( K_2 \) is supported on hyperplanes contributed from the union of the following four terms:
\begin{align}
    &\supp([UL_{1,X}U^\dagger, L_{2,X}]) \quad \supp([UL_{1,Z}U^\dagger, L_{2,X}]) \notag \\
    &\supp([UL_{1,X}U^\dagger, L_{2,Z}]) \quad\supp([U L_{1,Z}U^\dagger, L_{2,Z}]). \notag
\end{align}
\end{lemma}
\begin{proof}

    \begin{align*}
        \supp(K_2) &= \supp(K_1 L_{2,X}L_{2,Z}K_1^\dagger L_{2,Z}L_{2,X})\\
        &=\supp(K_1 L_{2,X}K_1^\dagger L_{2,X}L_{2,X}K_1 L_{2,Z}K_1^\dagger L_{2,Z}L_{2,X})\\
        &\subseteq\supp(K_1 L_{2,X}K_1^\dagger L_{2,X})\cup \supp(L_{2,X}K_1 L_{2,Z}K_1^\dagger L_{2,X} L_{2,Z}).
    \end{align*}
Conjugation by a 1-constant circuit like $L_{2,X}$ does not increase the support, so we get.
\begin{align*}
    \supp(K_2) &\subset \supp(K_1 L_{2,X}K_1^\dagger L_{2,X})\cup \supp(K_1 L_{2,Z}K_1^\dagger L_{2,Z})\\
    &=\supp([K_1, L_{2,X}]) \cup \supp([K_1, L_{2,Z}]).
\end{align*}

Furthermore, using $K_1 = (UL_{1,X}U^\dagger)(UL_{2,Z}U^\dagger)$ and noting that $[UL_{1,X}U^\dagger,UL_{2,Z}U^\dagger]= \pm 1$, we get: 
\begin{align*}
    &\supp(K_1 L_{2,X}K_1^\dagger L_{2,X})\\
    =& \supp((UL_{1,X}U^\dagger)(UL_{1,Z}U^\dagger)L_{2,X}(UL_{1,Z}U^\dagger)^\dagger(UL_{1,X}U^\dagger)^\dagger L_{2,X})\\
    \subseteq& \supp([UL_{1,X}U^\dagger,L_{2,X}])\cup \supp([UL_{1,Z}U^\dagger, L_{2,X}])\cup\supp({L_{2,X}}).
\end{align*}

$\supp(UL_{1,X}U^\dagger)$ is given by $\{\xi_i\}_i$ while $\supp(UL_{1,Z}U^\dagger)$ is given by $\{\mu_j\}_j$. Similarly, we also obtain: 
\begin{align*}
    &\supp(K_1 L_{2,Z}K_1^\dagger L_{2,Z})\\
    \subseteq& \supp([UL_{1,X}U^\dagger,L_{2,Z}])\cup \supp([UL_{1,Z}U^\dagger, L_{2,Z}])\cup\supp({L_{2,Z}})
\end{align*}
In our analysis, we will always include in the support of $L_{2,X}$ and $L_{2,Z}$, thus we will simply consider the support of $K_2$ as given by these four group commutators. 
\end{proof}

The first two terms is about the support of the intersection with $L_{2,X}$. Note that any hyperplane in $V_1$ will not intersect with $L_{2,X}$ (eliminating Case 1). We denote by $T \subset \{\xi_i\} \cup \{\mu_j\}$ as the subset of hyperplanes satisfying Case 3, where we can show that 
the maximum possible dimension of hyperplanes from intersection is $r-l$ and $0$ (as indicated in the table~\ref{fig: 4 cases of intersections}). For any $\gamma \in T$, consider $\supp{[\gamma]}\cap \supp[L_X]$. Since $\mathrm{dir}{(\gamma)}\subseteq \{1,\cdots, r\}$ and $\mathrm{dir}(L_X)=\{r+1,\cdots t\}$, $\{r+1\cdots t\}\cup \mathrm{dir}(\gamma)\subseteq \mathrm{dir}(\supp{[\gamma]}\cap \supp[L_X])$, which means that the size of this set is at least $t-r +l= 2l$. It is clear that if $\gamma \in \{\mu_j\},$ then $\mathrm{dir}(\gamma) = \{1,\cdots r\}$ so $\mathrm{dir}(\supp{[\gamma]}\cap \supp[L_X])= [t]$. For the last two terms, the third term only contributes $0$-dim hyperplanes which are individual physical qubits. For the last term, we note that $\mathrm{dir}{(\gamma)}$ must intersect non-trivially with $\mathrm{dir}(L_{2,Z})$, so it is an empty set. 

With our assumption that the constant depth circuit doesn't increase the dimension of hyperplanes, so $K_2$ must be supported on at most $(r-l)$-dim hyperplanes.

Next, we bound the total number of hyperplanes for $K_2$. We first show the following lemma.

\begin{lemma}
    Suppose $K_j = K_{j-1}L_jK_{j-1}^\dagger L_j$, with $L_j$ being some logical operator and $K_j$ defined recursively from $K_1 = UL_1U^\dagger$, then the number of hyperplanes supporting $K_j$ is given by $c \cdot (\text{number of hyperplanes for }L_j) \cdot (\text{number of hyperplanes for }K_{j-1}))$ and the dimension is the same as the intersection between the support of $L_j$ and $K_{j-1}$. 
\end{lemma}
\begin{proof}
    Let $L_1$ be a logical operator supported on a single hyperplane.
    $K_1 = UL_1U^\dagger$ is supported on $c$ hyperplanes, using the property of $U$. 
    \begin{align*}
        K_2 &= K_1 L_2 K_1^\dagger L_2^\dagger\\
        &= K_1 (L_2|_{\supp[K_1]})K_1^\dagger L_2^\dagger 
    \end{align*}
    \begin{align*}
        \supp[K_2] &\subseteq \supp[K_1 (L_2|_{\supp[K_1]})K_1^\dagger]\cup \supp[L_2]\\
        &= \supp[UL_1 U^\dagger(L_2|_{\supp[K_1]})U L_1 U^\dagger]\cup \supp[L_2]\\
    \end{align*}
    For the first term we get:
    \begin{align*}
        &\supp[U L_1U^\dagger(L_2|_{\supp[K_1]\cup \supp[L_1]})UL_1 U^\dagger]\\
        \subseteq& \supp[UL_1((U^\dagger L_2U)|_{\supp[K_1]\cup\supp[L_1] })L_1U^\dagger]\\
        \subseteq& \supp[L_2] \cup \supp[U(L_1U^\dagger L_2UL_1)|_{\supp[L_1]}U^\dagger]
        \end{align*}
This is equal to two terms: On the region outside of the support of $L_1$, we get a cancellation of $U$ and $U^\dagger$. The second term is from the intersection of the support of $U^\dagger L_2U$ with $L_1$ where $U$ acts non-trivially. 

If we bound the number of hyperplanes, then this is given by:
\begin{equation*}
    (c\cdot (\text{Number of hyperplanes in }L_2)) \cdot (\text{Number of hyperplanes in }L_1) \cdot c,
\end{equation*}
where the first term comes from the spread of $L_2$ under conjugation by $U$, and second term comes from intersection with $L_1$ and the third term comes from the last conjugation by $U$. To express this in terms of number of hyperplanes in $K_1$, we simply note that $\text{number of hyperplanes in } K_1=(\text{Number of hyperplanes in }L_1) \cdot c $:
\begin{equation*}
   (\text{Number of hyperplanes in }K_2)= (\text{Number of hyperplanes in }L_2) \cdot (\text{Number of hyperplanes in }K_1)\cdot c.
\end{equation*}
Further, since $U$ is dimension preserving, the dimension of hyperplanes in $K_2$ can be obtained from the dimension of the intersection between $K_1$ and $L_2$.

Next for the counting of the number of hyperplanes in $K_j= K_{j-1}L_j K_{j-1}^\dagger L_j^\dagger$. We first re-express the term.
\begin{align*}
    \supp[K_j] \subseteq  \supp[K_{j-1}L_j|_{\supp[K_{j-1}]}K_{j-1}^\dagger]\cup\supp[L_j].
\end{align*}
With the first term rexpressed as below:
\begin{align*}
&\supp[K_{j-1}L_j|_{\supp[K_{j-1}]}K_{j-1}^\dagger] \\
    =&\supp[(K_{j-2}L_{j-1}K_{j-2}^\dagger L_{j-1})L_j|_{\supp[K_{j-1}]}(L_{j-1}K_{j-2}L_{j-1}K_{j-2}^\dagger)]\\
    \subseteq& \supp[K_{j-2} \left[  L_{j-1}  \left(  K_{j-2} L_j K_{j-2}\right)|_{\supp[L_{j-1}]}  L_{j-1}\right]  K_{j-2}^\dagger] 
\end{align*}
Counting the number of hyperplanes from inner terms to outer terms:
\begin{align*}
    &(c \cdot (\text{number of hyperplanes in } L_j))\cdot(\text{Number of hyperplanes in } L_{j-1})\cdot c\\
    =& c \cdot (\text{number of hyperplanes in } L_j)) \cdot (\text{number of hyperplanes in } K_{j-1})
\end{align*}
On the first line, the first term is obtained from the spread of $L_j$ by $K_{j-2}$, the second term is obtained from the intersection with $L_{j-1}$ and the last term is obtained from the outermost conjugation by $K_{j-2}$. The second line is obtained by induction hypothesis on the number of hyperplanes in $K_{j-1}$. 
\end{proof}

Using the above lemma, we now count the number of hyperplanes supporting $K_2$. First taking intersection between $L_{2,X}$ and $\{\xi_i\}\cup \{\mu_j\}$: the intersection with $c\kappa_{1,X}$ $r$-dim hyperplanes $\{\xi_i\}$ gives $(c\kappa_1\kappa_{2,X} w +\kappa_{2,X})\cdot (c\kappa_{1,X})$ at most $(r-l)$-dim hyperplanes, while the intersection with $c\kappa_{1,Z}$ $l$-dim hyperplanes $\{\mu_j\}$ gives $(c\kappa_1 \kappa_{2,X}w +\kappa_{2,X})\cdot (c\kappa_{1,Z}) $ at most $0$-dim hyperplanes. On the other hand, for the intersection between $L_{2,Z}$ and $\{\xi_i\}\cup \{\mu_j\}$, it yields $(c\kappa_1\kappa_{2,Z} w +\kappa_{2,Z}) \cdot c\kappa_{1,X}$ $0$ dimensional hyperplanes. 

In total, we obtain at most $(c\kappa_1\kappa_{2,X} w +\kappa_{2,X})\cdot (c\kappa_{1,X})\cdot c\leq c^2 (c\kappa_1^2\kappa_2 w + \kappa_2\kappa_1)$ $(r-l)$-dim hyperplanes and $c^3\kappa_1^2\kappa_2 w + 2c^2(\kappa_1\kappa_2)$ $0$-dim hyperplanes after taking into account spreading to at most $c$ hyperplanes of the same dimension.

Calculations given below:
\begin{align*}
    &(c\kappa_1 \kappa_{2,X}w +\kappa_{2,X})\cdot (c\kappa_{1,Z})+(c\kappa_1 \kappa_{2,Z}w +\kappa_{2,Z}) \cdot c\kappa_{1,X}\\
    =&c^2\kappa_1^2\kappa_2 w + c(\kappa_{2,X}\kappa_{1,Z}+\kappa_{1,X}\kappa_{2,Z})\\
    \leq& c^2\kappa_1^2\kappa_2 w + 2c(\kappa_1\kappa_2)
\end{align*}

Next consider $K_3 = K_2 L_3 K_2^\dagger L_3^\dagger$. We will be able to obtain a nice recursive formulae after $K_3$ and the recursion ends when the support is on hyperplanes with dimension less than $l$.

As before, we can divide cases into 
\begin{align}
	& \mathrm{dir}(\xi) \cap \mathrm{dir}(L_X) \neq \emptyset, \quad
	\mathrm{dir}(\xi) \cap \mathrm{dir}(L_Z) \neq \emptyset, \\
	& \mathrm{dir}(\xi) \cap \mathrm{dir}(L_X) = \emptyset, \quad
	\mathrm{dir}(\xi) \cap \mathrm{dir}(L_Z) = \emptyset.
\end{align}
for hyperplanes $\xi$ in $K_2$. The analysis is now same to the previous case. We start with the canonical logical representatives for $L_3$ and deform using the results of Lemma~\ref{lemma: classical-cleaning} to form an alternative set of logical representatives. Since there are at most $\alpha= c^2\kappa_1\kappa_2 (3 + 2c\kappa_1 w)$ (i.e. total number of hyperplanes from $K_2$)values at each coordinate. Assuming $\frac{d-1}{2}>\alpha$, Lemma~\ref{lemma: classical-cleaning} can be applied. 

After cleaning, the new logical representatives for $L_{3,X}$ is supported on at most
\begin{equation*}
\alpha \kappa_{3,X}w +\kappa_{3,X}.
\end{equation*}
And the logical representatives for $L_{3,Z}$ is supported on at most:
\begin{equation*}
\alpha\kappa_{3,Z} w +\kappa_{3,Z}.
\end{equation*}

Taking intersection as before, and after spreading, $K_3$ is supported on at most 
\begin{align*}
&\mathrm{spread }\times \{\mathrm{\#\; of \; }L_{3,X}\}\times \{\# \mathrm{\; of \;} (r-l)\mathrm{-dim\; hyperplanes}\}\\
    =&c\cdot (\alpha \kappa_{3,X}w +\kappa_{3,X})\cdot (c^2 (c\kappa_1^2\kappa_2 w + \kappa_2\kappa_1))\\
    =&O(\kappa^7)
\end{align*}
$r-2l$-dim hyperplanes, and the rest of the terms do not intersect. The $l$-dim hyperplanes from $L_{3,Z}$ do not intersect with $K_2$ by construction, and the hyperplanes from $L_{3,X}$ can only intersect with $(r-l)$-dim hyperplanes if $r-l\geq l$. 

Next, we can recursively do dimension reduction: If $K_{i}$ is supported on at most $\alpha(\kappa_{\max\{[i]\}})$ $(r-(i-1)l)$-dim hyperplanes, then $K_{i+1}$ is supported on at most $\alpha^2(\kappa_{\max\{[i+1]\}})\cdot \kappa_{i+1}\cdot c$ $(r-l\cdot i)$-dim hyperplanes. We end the recursion when $r-l\cdot i<l$.

In summary, we will need $\frac{d-1}{2}>$ for distance.
The number of hyperplanes of $<l$ dim in the end is given by
\begin{equation*}
    O(\kappa_{\max\{[\lfloor\frac{r}{l}\rfloor+1]\}}^{3^{\lfloor\frac{r}{l}\rfloor}}).
\end{equation*}
The results are valid only if we can perform cleaning according the Lemma~\ref{lemma: classical-cleaning} which requires the minimum distance $d = \min \{d_i, d_j^T\}, 1\leq i\leq r, r+1\leq j\leq t $ of the code to scale as
\begin{equation*}
    (d-1)/2> O(\kappa_{\max\{[\lfloor\frac{r}{l}\rfloor+1]\}}^{3^{\lfloor\frac{r}{l}\rfloor-1}}).
\end{equation*}
Assume $\kappa_i$ is a constant, then this is $O(1)$ number of hyperplanes with dimension less than $r$. From here, we can recursively deduce that $U$ must implement a logical gate in the $\lfloor r/l\rfloor +1$-th level of the Clifford hierarchy. 
\end{proof}

We now show that any $O(1)$ set of $<l$ dimensional hyperplanes forms a correctable subset of physical qubits. A similar proof for hypertubes is given in Lemma \ref{lemma: kunneth-hpc-correctable}.

\begin{lemma}\label{lemma:correctable_hyperplanes}
    Let $d=\min \{d_i,d_i^T\}$, where $1\leq i\leq t$, and let $R$ denote a region consisting of $<d$ hyperplanes of (possibly different) orientations with dimension at most $l-1$ (assuming $l<r$) supported on a single sector,  then $R$ is a correctable region. 
\end{lemma}

\begin{proof}
We prove the claim by induction on the dimension $l$ of a canonical logical $Z$ representative which may be supported on multiple hyperplanes. 

As before, let $t$ be the dimension of the hypergraph product code and let the canonical $X$ and $Z$ logical representatives be supported on $r$ and $l$ dimensional hyperplanes respectively, so that $l+r =t$. Without loss of generality, We focus on a single sector. The case for multiple sectors easily follows by simply verifying for each logical operator that can be supported on a single sector. Let the canonical coordinates specifying the X-type hyperplanes be $(x_{r+1},\cdots x_{l+r =t})$ and those specifying the $Z$ type hyperplanes be $(x_1,\cdots, x_r)$.

We prove by contradiction. Assume there is a nontrivial $Z$ logical operator that can be supported on only $O(1)$ many $l-1$-dim hyperplanes, and that its canonical logical representative is fully supported on a single sector. Then, we know that these two representatives must be equivalent up to some $Z$ stabilizers. Thus, the idea of the proof is to start with the canonical logical representative and show that it is not possible using any stabilizer deformation, thereby ruling out a nontrivial $Z$ logical operator supported on only $O(1)$ hyperplanes of dimension at most $l-1$. 

Base case $(r=1)$: The canonical $Z$ logical representatives are supported on 1-dim hyperplanes, and we are given $O(1)$ 0-dimensional hyperplanes. If this is less than code distance $d$, then any $O(1)$ 0-dimensional hyperplanes (qubits) must form a correctable region.

Inductive step: Assume that the statement is true for $l = s_0$: Any $O(1)$ $s_0-1$-dim hyperplanes is correctable. We show that any  $Z$ logical operator supported on some $l = s_0+1$-dim hyperplanes cannot be supported on $O(1)$ $s_0$-dim hyperplanes. 

Recall that canonical logical representative in a hypergraph product code have product structure: the canonical logical $Z$ operator corresponds to a tensor product of classical codewords along the $l$ axes, with fixed values on the complementary $r$ coordinates. Fixing the coordinates \( (x_{r+1}, \dots, x_{t - s_0}) \), the restriction of the canonical logical \( Z \) operator to the remaining \( s_0 \) directions \( (x_{t - s_0 + 1}, \dots, x_t) \) defines an \( s_0 \)-dimensional codeword (with some thickness) arising from the product of the corresponding \( s_0 \) classical codes. 

Now consider deforming this representative with the $Z$ stabilizer generators. The generators are those naturally given by the classical codes $A_i^T$. These generators, when restricted to one sector, form sets of 1-dim lines on the $t$-dim lattice, with each line supporting a copy of stabilizers from $A_i^T$. Each generator thus acts on a 1-dim line aligned along a coordinate axis that is perpendicular to the $l$ coordinates that support a canonical $Z$ logical representative. Thus, using this set of stabilizer generators, the codeword can only be deformed within a $(s_0+r)$-dim hypercube. This consideration is general enough, in that we don't need to consider other $Z$ stabilizers that do not have the same values on the coordinates $(x_{r+1},\cdots, x_{t-s_0})$, which will only increase the support of the deformed $1$-dim codeword. 

At this point there are two possibilities:
\begin{enumerate}
    \item We can support this $s_0$-dim codeword with a $s_0$-dim hyperplane that is fully supported on this hypercube.
    \item Otherwise, all the $s_0$-dim hyperplanes are arranged such that only at most $s_0-1$-dim hyperplanes are fully supported on this hypercube.
\end{enumerate}
In case (2), the induction hypothesis implies that this will require at least $d > O(1)$ $(s_0-1)$-dim hyperplanes which will lead to a contradiction, so only case (1) can hold. 

For each of the coordinates $(x_{r+1},\cdots, x_{t-s_0})$, we have a disjoint $s_0+r$-dim hypercube, which requires at least one $s_0$-dim hyperplane, so this will require at least $d^{t-s_0-r}$ $s_0$-dim hyperplanes. Hence, we cannot support a logical $Z$ with constant number of $s_0$-dim hyperplanes.
\end{proof}

\subsection{Constant-depth circuits across sectors}\label{subsection: constant-depth-general-case}

In known practical context, such as the recent construction of logical $\CCZ$ operators for the hypergraph product codes~\cite{golowich2024quantumldpccodestransversal, Breuckmann2024Cups, Lin2024transversal}, the application of constant-depth circuits typically acts across different sectors. We now show that a similar bound can be applied in this most general setting. A key distinction in this subsection is that we will deal with stabilizer deformations that would introduce undesirable support across all sectors, and thus the proof techniques developed for this setting are different from the previous cases. 
We take inspiration from the notions of \emph{robustness} or \emph{bipuncturability} of the classical parity check matrices from~\cite{burton2020}. More precisely, we show that in Lemma~\ref{lemma: information set}, a classical parity check matrix $A^{(i)}$ is always bipuncturable if $\min(d_i. d^T_i) > \max(k_i, k^T_i)$. The proof technique in this section is thus of independent interest, which can also be seen as an alternative proof for the previous subsections.

In what follows, we consider homological products of classical parity checks whose distance is much larger than the code dimension. That is, given $\mathcal{A}^{(i)}$, 
we require that $k_i / d_i = o(1)$ and $k^T_i / d^T_i = o(1)$. A famous example is the repetition code where the transpose code equals to itself with code parameters $[n_i, 1, n_i-1]$. The tensor product of this repetition code gives toric codes in general $t$-dimensions. 
Let the homological product codes constructed from these classical codes, then we impose the following assumption, for any $i \in [t]$ 
\begin{align}\label{eq: k-d-scaling-classical}
    \max(k_i, k_i^T)  = o(\min(d_i, d^T_i)).
\end{align}
We need this assumption to guarantee that we can always find a set of punctures that is disjoint from the previous set(s). It is possible that for specific codes, one can relax the assumption, that is when $k > d$, one might still be able to find a new set of punctures.   

In Section \ref{sub: position}, we colloquially say that a set of hyperplanes with the same orientation is a hypertube. We now define it formally for further use.

\begin{definition}[Hypertubes]\label{def:hypertube} 
    An \emph{$r$-dimensional hypertube} $\bvec{\xi}$ is represented by a vector in $\mathcal{F}_r(\mathcal{A})$ such that there exists a sector $S_\mu$ with precisely r coordinate indices, $i_1, \cdots ,i_r \in [t]$, satisfying $|\operatorname{res}_\mu(\bvec{\xi})_{i_a}| \geq \min(d_i, d^T_i)$ while for the other sectors $S_\nu \neq S_\mu$, there are at most $r$ coordinate indices $j_1, \cdots, j_r$ such that $|\operatorname{res}_\nu(\bvec{\xi})_{j_b}| \geq \min(d_{j_i}, d^T_{j_i})$. 
\end{definition}

We similarly define the \emph{orientation} or \emph{direction} of the hypertube $\bvec{\xi}$ in sector $S_\mu$ to be $\operatorname{dir}_\mu(\bvec{\xi}) = b_\mu$ where $b_\mu$ is the set of $t - r$ coordinate indices such that $|\operatorname{res}_\mu(\bvec{\xi})_i| < \min(d_i, d^T_i)$. As a reminder, we also treat the hypertube as the set of coordinates with nonzero components in $\bvec{\xi}$. It is a thin set of hyperplanes normal to the direction $b_\mu$.

\begin{definition}[Support map] 
	We denote the support map of $U$ by $\Phi_U$ such that for any vector $\bvec{x} \in \mathcal{F}_r (\mathcal{A})$
    \begin{align}
        \Phi_U(\bvec{x}) := \operatorname{supp}(\bvec{x}) \cup \operatorname{supp}(U\bvec{x}U^\dagger). 
    \end{align}
    Again, the right-hand side can be either understood as the union of supporting sets or as vectors in $\mathcal{F}_r (\mathcal{A})$ with nonzero components indicating the support.
\end{definition}

\begin{definition}[Shape-preserving constant-depth circuit] \label{def: shape-preserving}
	We say that the logical operator $U$ is \emph{shape-preserving} if for each $r$-dimensional hypertube $\bvec{\xi}$, 
	\begin{align}
		\Phi_U(\bvec{\xi}) = \bigcup_{\sigma \in S_t} \Phi_{U, \sigma}(\bvec{\xi}),
	\end{align}
	where $\Phi_{U, \sigma}(\bvec{\xi})$ is a supporting set defined as follows. We let $S_t$ be the permutation group on the index set $[t]$.
	For each sector $S_\mu$, $\Phi_{U, \sigma}(\bvec{\xi})$ has coordinate components satisfying $|\operatorname{res}_\mu(\Phi_{U, \sigma}(\bvec{\xi}))_{i} | =  O(|\operatorname{res}_\mu(\bvec{\xi})_{\sigma^{-1}(i)} |)$ for any $i \in [t]$. 
\end{definition}

In other words, the weight of nonzero components indexed by $i$ is bounded by that of $\operatorname{res}_\mu(\bvec{\xi})$ indexed by $\sigma^{-1}(i)$. With the control on weights, $\Phi_{U, \sigma}(\bvec{\xi})$ is still an $r$-dimensional hypertube with a similar shape as $\bvec{\xi}$. Moreover, we precisely know that it has the orientation $\mathrm{dir}_\mu(\Phi_{U, \sigma}(\bvec{\xi})) = \sigma(b_\mu)$. Finally, the union $\Phi_U(\bvec{\xi})$ contains all possible $\Phi_{U, \sigma}(\bvec{\xi})$'s with nonzero components being controlled by different coordinates. Note that since $t$ is a constant, the union over $S_t$ is still a (large) constant.

\begin{remark}
	The assumption in Eq.~\eqref{eq: k-d-scaling-classical} implies that $k_i, k^T_i$ as subleading terms compared with the distance of the classical parity checks. Intuitively, definition \ref{def: shape-preserving} says that applying $U$ to a hypertube may permute its orientation and send it to new seectors but in a controlled way with similar thickness in each sector and orientation. Thus, the shape-preserving condition allows us to treat these subleading terms in a similar footing as $O(1)$ contributions by the constant-depth circuits. The proof can be straightforwardly adapted to the cases without the shape-preserving condition, where we require more stringent rate and distance scaling from classical parity checks, namely, 
\begin{align}
    \max(k_i, k_i^T)^{g(t, l)} = o(\min(d_i, d^T_i))
\end{align}
where the exponent may depend on the dimensionality $t$ of the chain complex as well as the level $l$ on which we place our physical qubits. 
\end{remark}

It is useful to write a more explicit form of hyperplanes. We define the coordinate form of any hyperplane. Taking the homological products of classical parity checks $\mathcal{A}$ with $\mathcal{A}^{(1)} \otimes \cdots \otimes \mathcal{A}^{(t)}$. We denote a \emph{canonical coordinate} as a Cartesian product of index sets $[i_1, i_2, \cdots, i_t] \in \mathcal{I}_{l}$, where $\mathcal{I}_l$ contains all the tensor product of vector spaces (sectors) associated with $\mathcal{F}_l[\mathcal{A}]$. As in Section \ref{sub: position},
we fix one sector by $S_\mu$ 
For example, suppose that $\mu = \mathcal{A}^{(1)}_{1} \times \cdots \times\mathcal{A}^{(l)}_{1} \times \mathcal{A}^{(l+1)}_{0} \times \cdots \times \mathcal{A}^{(t)}_{0}$. Then we say the canonical coordinate restricted onto $S_\mu$, 
\begin{align}
    \operatorname{res}_\mu(\mathcal{I}) \cong \mathbb{F}^{n_1}_2 \otimes \cdots \otimes \mathbb{F}^{n_l}_2 \otimes \mathbb{F}^{m_{l+1}}_2 \cdots \otimes \mathbb{F}^{m_t}_2.
\end{align}

\begin{definition}[Simple vector]
    Any vector $\bvec{x} \in \mathcal{F}_l[\mathcal{A}]$, $\bvec{x} = [x_1, \cdots, x_t] \in \mathcal{I}_l$ is said to be \emph{simple} if when restricts to any $S_\mu$, $\mathrm{res}_{\mu}(\bvec{x})$ is a simple tensor product. We define $\mathrm{res}_\mu (\bvec{x})_i = \mathrm{res}_\mu(x_i)$ its $i$-th component in the tensor product.
\end{definition}

\begin{example}
    Consider a $2$-dimensional hypertube on a $3$-dimensional hypergraph product code $\bvec{\xi} = [\xi_1, \xi_2, \xi_3]$ such that it is only nonvanishing at the sector $S_1 = \mathbb{F}^{m_1 \times m_2 \times n_3}_2$  where $|\operatorname{res}_1(\bvec{\xi})_1| = |\operatorname{res}_1(\xi_1)| \geq d^T_1$ and $|\operatorname{res}_1(\bvec{\xi})_2| = |\operatorname{res}_1(\xi_2)| \geq d^T_2$ and   $|\operatorname{res}_1(\bvec{\xi})_3| = |\operatorname{res}_1(\xi_3)| =1$. Take $\sigma = (1 23)$. Then $\Phi_{U, \sigma}(\bvec{\xi}) = [\xi'_1, \xi'_2, \xi'_3]$ where under any sector $S_1, S_2, S_3$, the length of $\operatorname{res}_\mu(\xi'_i)$ is amplified by at most a constant factor of the length of $\operatorname{res}_1(\xi_{\sigma^{-1}(i)})$, as a consequence of shape-preserving assumption on $U$. 
\end{example}

\begin{example}[Cover for logical operators] 
	Following the notations in the above example, we consider the product of all $X$ and $Z$ logical operators via the canonical representation $L^X_\mu$ and $L^Z_\mu$. $L^X_\mu$ and $L^Z_{\mu}$ need not be simple; however, we could define their cover as some hypertube. Choose a $r$-dimensional hypertube $\xi = (\xi_1, \cdots, \xi_{l+1}, \cdots, \xi_t)$. For $i \in \{1, \cdots, l\}$, we have $\xi_i$ to be a vector containing all the puncture index $\gamma_i$ in $\ker A^{(i)}$. For  $i \in \{l+1, \cdots, t\}$, we have $\xi_i$ to be an all-one vector. Then it follows that $\operatorname{supp} (\xi) \supset \operatorname{supp} L^X_\mu$ and $\xi$ is a $r$-dimensional hypertube. 
\end{example}

\begin{lemma}\label{lemma: disjoint-intersection-hyperlanes}
	Let $s < t$ be an integer and $\bvec{\xi}$ be a $s$-dimensional hypertube and let $S_\mu$ be a sector from $\mathcal{F}_l(\mathcal{A})$. Denote its orientation on $S_\mu$ by $\operatorname{dir}_\mu(\bvec{\xi})$. For any logical operators $L^X_\mu (L^Z_\mu)$ supported on exclusively on $S_\mu$, if $\operatorname{dir}_\mu(L^X_\mu) \cap \operatorname{dir}_\mu(\bvec{\xi}) \neq \emptyset$ ($\operatorname{dir}_\mu(L^Z_\mu) \cap \operatorname{dir}_\mu(\bvec{\xi}) \neq \emptyset$), then there must exists a choice of basis such that their supports do not overlap. 
\end{lemma}
    \begin{proof}
	We prove the case for $L^X_\mu$ and the argument for $L^Z_\mu$ follows analogously. Let's denote $j \in [t]$ such that $j \in \operatorname{dir}_\mu(L^X_\mu) \cap\operatorname{dir}_\mu(\bvec{\xi})$. Recall that in this case, we have that $|\operatorname{supp} (\operatorname{res}_\mu(\bvec{\xi})_j)| < \min(d_j, d^T_j) $. Thus, we can always choose the puncture set $\gamma_j$ as an information set associated with $\ker A_j $ such that $\gamma_j \subset [t] \setminus \operatorname{supp} (\operatorname{res}_\mu(\bvec{\xi})_j)$. Recall that $\bvec{\xi}$ and $L^X_\mu$ are simple vectors so that 
	\begin{align}
		\begin{aligned}
			&\operatorname{supp} (\bvec{\xi}) \cap \operatorname{supp} (L^X_\mu )= \operatorname{supp} (\operatorname{res}_\mu \bvec{\xi} )\cap \operatorname{supp} (L_\mu) \\
			&= \operatorname{supp} (\operatorname{res}_\mu (\bvec{\xi})_1 ) \cap \operatorname{supp} (\operatorname{res}_\mu(L^X_\rho)_1) \times \cdots \times \operatorname{supp} (\operatorname{res}_\mu (\bvec{\xi})_t)  \cap \operatorname{supp} (\operatorname{res}_\mu(L^X_\rho)_t). 
		\end{aligned}
	\end{align}
	Since $\operatorname{supp} (\operatorname{res}_\mu (\bvec{\xi})_j ) \cap \operatorname{supp} (\operatorname{res}_\mu(L^X_\rho)_j) = \emptyset$, it follows that their supports do not overlap. 
\end{proof}

\begin{lemma}
Let the physical qubits be encoded in $\mathcal{F}_l(\mathcal{A})$ where $l \leq r $ and $r + l = t$. For any $s$-dimensional hypertubes $\bvec{\xi}$ such that $s \leq l$, there exists a choice of logical basis sets such that any logical operators would have an intersection with $\bvec{\xi}$ that is either disjoint or has an intersection region $[\delta_1, \cdots, \delta_t]$ where $|\operatorname{res}_\rho(\delta_i)| < \min(d_i, d^T_i)$  for any sector $S_\rho$. 
    \begin{proof}
        Without loss of generality, let's assume that the support of $\bvec{\xi}$ is written in the following coordinate form 
        \begin{align}
            \bvec{\xi} = [\xi_1, \cdots, \xi_{s}, e_{s+1}, \cdots, e_t]
        \end{align}
        where $|\operatorname{res}_\mu(\xi_i)| \geq \min(d_i, d^T_i) $ if its restriction onto $S_\mu$ is non-trivial and $|\operatorname{res}_\mu(e_i)| < \min(d_i, d^T_i) $. Let's suppose that $s =  l$ and denote a sector $S_\mu \cong \mathbb{F}^{n_1 \times \cdots \times n_l}_2 \times \mathbb{F}^{m_{l+1} \times \cdots \times m_t}_2 $. Denote the logical operator $L^X_\mu$ supported on $S_\mu$. It follows that their intersection of supports is given by 
        \begin{align}
           \operatorname{supp}  \bvec{\xi} \cap \operatorname{supp} L^X_\mu =  \vartheta_1 \times \cdots \times \vartheta_t \subset S_\mu
        \end{align}
        where  $|\operatorname{res}_1(\vartheta_i) | \leq \max(k_i, k_i^T)$. In the case where $l < r $ and consider logical operators $L_\rho$ supported on any sectors $S_\rho$ with $\rho \neq \mu$. Then we have that 
        \begin{align}
            \begin{aligned}
                &\operatorname{supp} \bvec{\xi} \cap \operatorname{supp} L_\rho = \operatorname{supp} \operatorname{res}_\rho \bvec{\xi} \cap \operatorname{supp} L_\rho \\
            =& \operatorname{supp} \operatorname{res}_\rho (\bvec{\xi})_1  \cap \operatorname{supp} \operatorname{res}_\rho(L^X_\rho)_1 \times \cdots \times \operatorname{supp} \operatorname{res}_\rho (\bvec{\xi})_t  \cap \operatorname{supp} \operatorname{res}_\rho(L^X_\rho)_t    \\
            & \bigcup \operatorname{supp} \operatorname{res}_\rho (\bvec{\xi})_1  \cap \operatorname{supp} \operatorname{res}_\rho(L^Z_\rho)_1 \times \cdots \times \operatorname{supp} \operatorname{res}_\rho (\bvec{\xi})_t  \cap \operatorname{supp} \operatorname{res}_\rho(L^Z_\rho)_t
            \end{aligned}
        \end{align}
        where for the second equality we utilize the fact that $\bvec{\xi}$ is a simple vector. Then there must exist $j \in [t]$ such that  $\operatorname{supp} \operatorname{res}_\rho (\bvec{\xi})_j  \cap \operatorname{supp} \operatorname{res}_\rho(L^X_\rho)_j = \operatorname{supp} \operatorname{res}_\rho(e_j) \cap \operatorname{supp} \operatorname{res}_\rho(u_{\gamma_j}) $ where $u_{\gamma_j} \in \mathbb{F}^{n_j}_2$ is nonzero on information set $\gamma_j$ associated with $\ker A^{(j)}$. By Lemma~\ref{lemma: information set}, since by assumption we have that $|\operatorname{supp} \operatorname{res}_\rho(e_j)| < \min(d_j^T, d_j)$. Then we can always choose $\operatorname{supp} e_{\gamma_j} \subset \overline{\operatorname{supp} e_j}$ so that the intersection $\operatorname{supp} \bvec{\xi} \cap \operatorname{supp} L^X_\rho = \emptyset$. Similarly we argue that $\operatorname{supp} \bvec{\xi} \cap \operatorname{supp} L^Z_\rho = \emptyset$. Hence its intersection of any logical operator must be contained by $ \operatorname{supp}  \bvec{\xi} \cap \operatorname{supp} L^X_\mu = \vartheta_1 \times \cdots \times \vartheta_t \subset S_\mu$
        
        In the case where $r=l$, the by symmetry, we denote $L^Z_\nu$ supported on the sector $S_\nu \subset \mathbb{F}^{m_1 \times \cdots \times m_r}_2 \times \mathbb{F}^{n_{r+1} \times \cdots \times n_t}_2$ and 
        \begin{align}
            \operatorname{supp} \bvec{\xi} \cap \operatorname{supp} L^Z_\nu = \operatorname{supp} \operatorname{res}_\nu(\bvec{\xi}) \cap \operatorname{supp} L^Z_\nu =  \vartheta'_1 \times \cdots \times \vartheta'_d
        \end{align}
        where  $|\operatorname{res}_2(\vartheta'_i) | \leq \max(k_i, k_i^T)$. Similar argument shows that any logical $L$, we can always choose the information sets such that the intersection $\operatorname{supp} \bvec{\xi} \cap \operatorname{supp} L$ is contained by
        \begin{align}
            \operatorname{supp} \bvec{\xi} \cap \operatorname{supp} L^X_\mu \cup \operatorname{supp} \bvec{\xi} \cap \operatorname{supp} L^Z_\nu.
        \end{align}
        Note that $S_\mu$ and $S_\nu$ are from different sectors, we write this intersection by  
        \begin{align}
            [\delta_1, \cdots \delta_t] = ( \vartheta_1 , \cdots ,\vartheta_d)|_{S_\mu} \bigsqcup ( \vartheta_1 , \cdots ,\vartheta_d)|_{S_\nu}
        \end{align}
        where $|\operatorname{res}_\rho(\delta_i)|, |\operatorname{res}_\rho(\delta_i)| < \min(d_i, d^T_i)$ for any sector $S_\rho$, as desired. 
    \end{proof}
\end{lemma}

Recall for the surface code, any $l \times l$ region away from the boundary where $l < d$ is correctable. We wish to generalize this result for the homological product code. We say that a $l$-dimensional hypertube $\bvec{\xi}$ is \emph{dimension-preserving} if for any $h \in \mathcal{F}_{l+1}[\mathcal{A}]$  and for any sector $S_\mu$, we have that $\operatorname{res}_\mu(\bvec{\xi} + \partial h)$ can either $(i)$ be written as disjoint union of constant many hypertubes supported on $S_\mu$, $\operatorname{res}_\mu(\bvec{\xi} + \partial h) = \sqcup_{i} \zeta_{i, \mu}$ such that there exists one $\zeta_{i, \mu}$, $ |\operatorname{dir}_\mu(\bvec{\xi})| = |\operatorname{dir}_\mu(\zeta_{i, \mu} )|$. Or $(ii)$ be written as $\Omega(\min_{i}(\{d_i, d^T_i\}))$ many hypertubes of lower dimensions. Then we have the following result analogous to Lemma \ref{lemma:correctable_hyperplanes} for hypertubes. 

\begin{lemma}\label{lemma: kunneth-hpc-correctable}
    Any finite collection of $s < \min(l, r)$-dimensional hypertubes is correctable. 
    \begin{proof}
        A key observation we will use is that the canonical logical representatives are of ``minimal weight" and any stabilizer deformation preserves dimension. 
        \begin{claim}
            Let $\bvec{\xi}$ be $l$-dimension hypertube and let $\operatorname{res}_\mu(\bvec{\xi})$ be its restriction on any sector $S_\mu$. Suppose that the nontrivial supports of $\operatorname{res}_\mu(\bvec{\xi})$ on $\operatorname{dir}_\mu(\bvec{\xi})$ is defined exclusively on the information bits associated with $\ker A^{(i)}$ ($\ker A^{(i) T}$) for $i \in \operatorname{dir}_\mu(\bvec{\xi}) $, then $\bvec{\xi}$ is dimension-perserving.
            \begin{proof}
            It suffices to restrict our attention to any single sector $S_\mu$. Let's assume WLOG that the $\operatorname{dir}_\mu(\bvec{\xi}) = \{1, 2, \cdots, r\}$ where $r + l =t$. Let's first assume that $h = h^{(1)} \otimes \cdots \otimes h^{(t)}$. Then we have that, restricted on $S_\mu$, the stabilizer action, 
            \begin{align}
                A^{(1)}h^{(1)} \otimes \cdots \otimes A^{(r)}h^{(r)}\otimes  h^{(r+1)} \otimes \cdots \otimes h^{(t)}
            \end{align}
            Note that on the coordinates $i=1, \cdots, r$, $\operatorname{res}_\mu(\bvec{\xi})$ is supported on the information bits. Write, 
            \begin{align}
                \operatorname{res}_\mu(\bvec{\xi}) = \xi^{(1)} \otimes \cdots \otimes \xi^{(t)}
            \end{align}
            Note that $\xi^{(i)} \in \operatorname{coker} A^{(i)}$, then there cannot be any $h^{(i)}$ such that $A^{(i)} h^{(i)} = \xi^{(i)}$. Hence denote $v^{(i)}$ that lies outside of the information bits and disjoint from $\xi^{(i)}$ such that $A^{(i)} h^{(i)} = \xi^{(i)} + v^{(i)}$, then we can write the hypertube
            \begin{align}
                v^{(1)} \otimes \cdots \otimes v^{(r)} \otimes h^{(r+1)} \otimes \cdots \otimes h^{(t)},
            \end{align}
            which is disjoint from $\operatorname{res}_\mu(\bvec{\xi})$. It is then easy to see that in this case, then $\operatorname{res}_\mu(\bvec{\xi} + \partial h)$ is given by the disjoint union of two hypertubes, 
            \begin{align}
                 v^{(1)} \otimes \cdots \otimes v^{(r)} \otimes h^{(r+1)} \otimes \cdots \otimes h^{(t)} \bigsqcup  \xi^{(1)} \otimes \cdots \otimes \xi^{(r)} \otimes \xi^{(r+1)} + h^{(r+1)} \otimes \cdots \otimes \xi^{(t)} +h^{(t)}
            \end{align}
            \end{proof}
            Then it is straightforward to observe at least one of such hypertube is $l$-dimensional, which implies dimension-perserving. Notice that we can generalize to any $h \in \mathcal{F}_{l+1}[X]$ by linearity. 
        \end{claim}
        Using the claim, it follows that any canonical logical basis is dimension-preserving. Similar methods can be extended to proving for general logical operators, which can be written as a disjoint union of $\min(l, r)$-dimensional hypertubes. Then it follows that any finite collection of $s < \min(l, r)$-dimension hypertubes cannot contain logical operators, even restricted in any single sector.  
    \end{proof}
\end{lemma}

\begin{remark}
    It is helpful to consider $K_1 = UL U^\dagger $ where $L$ is a general logical operator where is given by the linear combination of different sectors. Then we in general can write 
    \begin{align}
        \operatorname{supp} K_1 \subseteq  \bigcup_{\sigma} [e'_{\sigma(1)}, \cdots, e'_{\sigma(l)}, \xi'_{\sigma(l+1)}, \cdots, \xi'_{\sigma(t)}].
    \end{align}
    Then for each sector $S_\mu$ the support of $K_1$ restricted onto $S_\mu$, $\operatorname{res}_\mu(K_1)$ has $m$-dimensional hypertubes at different orientations, which covers the scenario $1$. 
\end{remark}

\begin{lemma}\label{lemma: dimension-reduction-hyperplanes}
    For any sector $S_\nu$, there exists a logical basis set $L_\mu$ such that $\Phi_U(L_\mu)$ intersects with any logical operator $L$ must be given by a finite set of at most $s \leq r-l$-dimensional hypertubes. 
    \begin{proof}
        By definition, we can expand $\Phi_U(L_\mu)$ across different sectors, 
        \begin{align}
            \Phi_U(L_\mu) \subseteq \bigcup_{\sigma}  \Phi_{U, \sigma}(L_\mu)
        \end{align}
       We can decompose any logical operator $L$ into its components in each sector $S_\nu$, 
        \begin{align}
            \begin{aligned}
                &\Phi_U(L_\mu) \cap L \subseteq \left( \cup_{\sigma} \Phi_{U, \sigma}(L_\mu) \right)  \cap \left( \cup_{\nu} \operatorname{supp} L_\nu \right) = \bigcup_{\sigma, \nu } \Phi_{U, \sigma}(L_\mu) \cap \operatorname{supp} L_\nu  \\
            &= \bigcup_{\nu }  \operatorname{res}_\nu\left(\cup_{\sigma} \Phi_{U, \sigma}(L_\mu) \right)  \cap \operatorname{supp} L_\nu
            \end{aligned} 
        \end{align}
        
        Writing the support of $L_\nu$ to be 
        \begin{align}
            \operatorname{supp} L_\nu \subset   \operatorname{supp} L^X_\nu \cup  \operatorname{supp} L^Z_\nu
        \end{align}
        where we take $S_\nu \cong \mathbb{F}^{n_1 \times \cdots n_l}_2 \times \mathbb{F}^{m_{l+1} \times \cdots m_t}_2$ and $\operatorname{dir}_\nu(L^X_\nu) = \overline{\operatorname{dir}_\nu(L^Z_\nu)}$ without loss of generality. For any fixed $\sigma$, $\Phi_{U, \sigma}(L_\mu) \subset \Phi_{U, \sigma}(L^X_\mu) \cup \Phi_{U, \sigma}(L^Z_\mu)$ where $\Phi_{U, \sigma}(L^X_\mu)$ and $\Phi_{U, \sigma}(L^Z_\mu)$ are respectively $r$ $(l)$-dimensional hypertubes. Recall that $\Phi_{U, \sigma}(L^X_\mu)$ is a simple vector and we denote its orientation on $S_\nu$ by $\operatorname{dir}_{\nu}(\mu, \sigma)$. Suppose that there exists $j \in [t]$ such that $j \in \operatorname{dir}_{\nu}(\mu, \sigma) \cap \operatorname{dir}_\nu(X, \mu)$, then by Lemma~\ref{lemma: disjoint-intersection-hyperlanes}, then their supports are disjoint. Furthermore, for any $j \in \operatorname{dir}_\nu(X, \mu)$, consider the set of hypertubes oriented by $\sigma$ that contains $j$,
            \begin{align*}
                |\operatorname{supp} \operatorname{res}_\nu( \cup _{\sigma: j \in \operatorname{dir}_\nu(\mu, \sigma)}\Phi_{U, \sigma}(L^X_\mu))_j|  &\leq  |\cup _{\sigma: j \in \operatorname{dir}_\nu(\mu, \sigma)} \operatorname{supp} \operatorname{res}_\nu( \Phi_{U, \sigma}(L^X_\mu))_j | \\ &\leq c \max(k_j, k^T_j) < \min(d_j, d^T_j),
            \end{align*}
            where $c$ is some constant. Since the number of different orientations is given by at most $\binom{t}{l}$, which is a constant. Then there must exists a large (upto $t-1$-dimensional) hyperplane $\bvec{\zeta}$ such that, 
            \begin{align*}
                \operatorname{supp} \operatorname{res}_\nu( \cup _{\sigma: j \in \operatorname{dir}_\nu(\mu, \sigma)}\Phi_{U, \sigma}(L^X_\mu)) \subset \operatorname{supp} \operatorname{res}_\nu(\bvec{\zeta})
            \end{align*}
            such that $j \in \operatorname{dir}_{\nu}(\bvec{\zeta})$. Then by applying Lemma~\ref{lemma: disjoint-intersection-hyperlanes}, it follows that their supports must also be disjoint. Then we consider all the orientation of hypertubes such that $\operatorname{dir}_{\nu}(\mu, \sigma) \cap \operatorname{dir}_\nu(X, \mu) = \emptyset$. For each such orientation of hypertubes, the hypertubes from their intersection must have the orientation $|\operatorname{dir}_\nu(\mu, \sigma) \cup \operatorname{dir}_\nu(L_\nu^X)| \leq r-l$ as desired. A symmetry argument implies the same result for intersection between $\cup _{\sigma}\Phi_{U, \sigma}(L^X_\mu)$ and $L^Z_\nu$, and for the intersection between $\cup _{\sigma}\Phi_{U, \sigma}(L^Z_\mu)$ and $L^X_\nu$ ($L^Z_\nu$). 
    \end{proof}
\end{lemma}

We now proceed with dimensional reduction. For the base case where $ \lfloor \frac{r}{l} \rfloor =1$, we have $K_1 = UL_1 U^\dagger$ maps to $l$-dimensional hypertubes across all sectors. Consider that $L_1 = \prod_\mu L_\mu $ and $K_1 = \prod_\mu(UL_\mu U^\dagger)$. The support of $K_1$, $\Phi_U(L_1)$ can be written as a subset of 
\begin{align}
    \Phi_U(L_1) \subseteq  \bigcup_{\mu } \Phi_U(L_\mu). 
\end{align}

The above claim indicates that $\Phi_U(L_\mu)  \cap L_2$ contains a collection of at most $(r-l)$-dimensional hypertubes where $r-l < l$. It then follows that $\cup_\mu \Phi_U(L_\mu) \cap L_2$ contains another collection of at most $(r-l)$-dimensional hypertubes where $r-l < l$, where the size of the collections is constant. Furthermore, we have that $\Phi_U( \Phi_U(L_1) \cap L_2)$ which only extends the support to another constant-size collection of $r-l < l$-dimensional hypertubes. Hence, by Lemma~\ref{lemma: kunneth-hpc-correctable}, it is correctable. This proves for the base case. For the case where $\lfloor r / l \rfloor = g > 1$,  we can apply dimension reduction iteratively as follows. Recall that $K_{i+1} = K_i L_{i+1} K^\dagger_i L^{\dagger}_{i+1}$ where $\operatorname{supp} K_{i+1} = \Phi_{K_{i}}(\operatorname{supp} K_{i} \cap \operatorname{supp} L_{i+1}^{ \dagger} ) \subset \Phi_{U}(\operatorname{supp} K_{i} \cap \operatorname{supp} L^{\dagger}_{i+1} )$. By the inductive hypothesis where we denote $ \operatorname{supp}K_i$ is a constant-size collection of hyperplanes whose high dimensions are given by $r - (i-1)l \geq l$. Applying Lemma~\ref{lemma: dimension-reduction-hyperplanes} and arguments above, we reached the conclusion that $\operatorname{supp} K_{i+1}$ is given by a constant-size collection of hypertubes of dimension at most $r - i l$.  This allows us to conclude that 

\begin{theorem}
    Suppose for each classical code $\mathcal{A}^{(j)} $ for $j \in [t]$ whose code parameters satisfy Eq.~\eqref{eq: k-d-scaling-classical}. Then logical gates implementable by a constant-depth, shape-preserving circuit $U$ are restricted to $\mathcal{P}_{\lfloor r/l \rfloor+1}$, i.e., the (\( \lfloor r/l \rfloor+1 \))-th level of the Clifford hierarchy.
\end{theorem}

\section{Yes-go constructions and tightness} 
\label{sec: yes-go}

In this section, we discuss the constructions of fault-tolerant logical gates on hypergraph product codes and various canonical examples in both scenarios with and without geometric locality, with the aim of understanding whether and how our above no-go results of the Clifford hierarchy limits can be attained.

A recent framework suitable for our purpose utilizes a form of cohomology invariant called \emph{cup product} \cite{Chen_2023,Breuckmann2024Cups,Lin2024transversal}. In the context of $t$-dimensional simplicial complexes, a cup product $\smile: H^i(X) \times H^{j}(X) \rightarrow H^{i+j}(X)$ is a product between cohomologies defined by merging lower dimensional simplices $\xi \in C^i(X)$ and $\xi' \in C^j(X)$ as a higher dimensional one $\xi \smile \xi'$ and the cohomological classes can be preserved during the process.  

The idea of cup products can be more clearly illustrated by an analogous notion extended to the cubical complexes \cite{Chen_2023,Lin2024transversal}. They are made from Cartesian products of intervals, and have a closer relation with hypergraph products that we study in this work. For a hypergraph product, let us interpret each 1-chain $\mathcal{A}^{(i)} = (A^{(i)}: \mathbb{F}_2^{n_i} \rightarrow  \mathbb{F}_2^{m_i})$ as a mapping from the space $\mathbb{F}_2^{n_i}$ of $n_i$ intervals/edges to the space $\mathbb{F}_2^{m_i}$ of $m_i$ endpoints/vertices. Then elements in
\begin{align}
	\mathcal{D}_s = \bigoplus_{\vert [I] \vert = l} \bigotimes_{\substack{i \in [I] \\ j \in [t] \setminus [I]} } \mathbb{F}_2^{n_i} \otimes \mathbb{F}_2^{m_j} 
\end{align}
are visualized as $l$-dimensional cubes. 

\subsection{Topological codes}\label{sub:topo}

First consider the case of topological (geometrically local) codes, with high dimensional toric codes being the canonical example. Let $A_c: \mathbb{F}_2^n \rightarrow \mathbb{F}_2^n$ be a cyclic code. To distinguish edges and vertices we may adopt the notation $H_c: \mathbb{F}_2^{E} \rightarrow \mathbb{F}_2^{V}$. 
{For clarity, we discuss the $3$-dimensional (3D) case for now, but generalizations to higher dimensions are straightforward and will be discussed in the end.}
The 3D toric code is constructed from the hypergraph product of three copies of $H_c$:
\begin{align}\label{eq:3Dtoric}
	\hspace{-2mm}
	\mathbb{F}_2^{E_1} \otimes \mathbb{F}_2^{E_2} \otimes \mathbb{F}_2^{E_3} & \xleftarrow{\delta^2} 
	\bigoplus_{i \neq j \neq k} (\mathbb{F}_2^{E_i} \otimes \mathbb{F}_2^{E_j} \otimes \mathbb{F}_2^{V_k})	
	\xleftarrow{\delta^1} 
	\bigoplus_{i \neq j \neq k} (\mathbb{F}_2^{E_i} \otimes \mathbb{F}_2^{V_j} \otimes \mathbb{F}_2^{V_k})	
	\xleftarrow{\delta^0} 
	\mathbb{F}_2^{V_1} \otimes \mathbb{F}_2^{V_2} \otimes \mathbb{F}_2^{V_3},
\end{align}
where we take the cochain complex since cup products are defined on cohomologies. The quantum code is taken from the consecutive terms on the right with $\delta^1 = H_Z$ and $\delta^0 = H_X^T$, which is consistent with Eq.~(\ref{eq:HGP}) by taking transpose. The physical qubits are placed on 1D edges and are separated into three sectors. In one sector $\mathbb{F}_2^{E_i} \otimes \mathbb{F}_2^{V_j} \otimes \mathbb{F}_2^{V_k}$, a canonical $X$ logical representative is defined as a 2D hyperplane on one puncture (see Section~\ref{sub: generalHGP}). Pictorially, it consists of paralleled edges through the perpendicular 2-dimensional plane in the 3D cube as in Fig.~\ref{fig:cube-intersection}. 

\begin{figure}[ht]
    \centering
    \includegraphics[width=\textwidth]{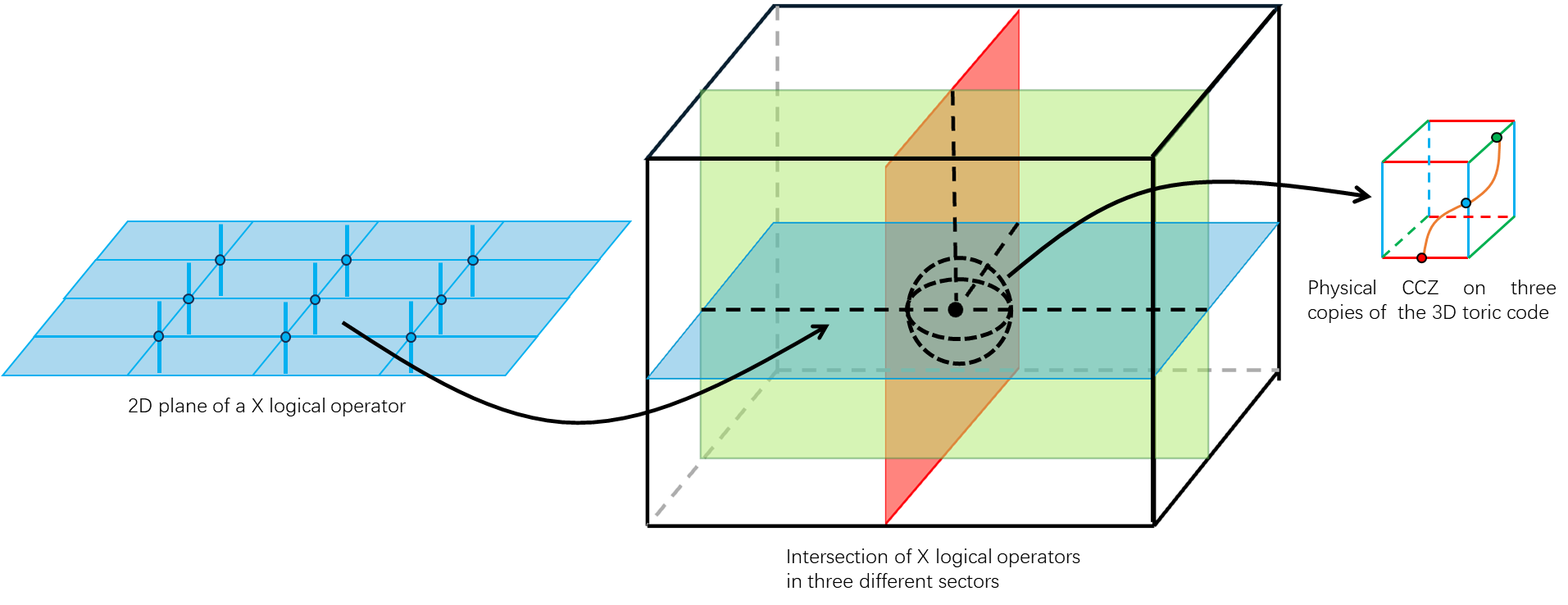}
    \caption{Intersection of logical $X$ operators for three copies of $3$-dimensional toric code from red, green, and blue, respectively. The intersection of the red and green planes is drawn with a black dashed line, which intersects perpendicularly to the blue plane. Such an intersection in logical operators indicates a nontrivial action of logical $\CCZ$ (which is realized through physical $\CCZ$s acting on six paths of each cube).}
    \label{fig:cube-intersection}
\end{figure}

Three different sectors support three independent $X$ logical representatives, corresponding to red, green, and blue planes in the cube, respectively. It is easy to observe that the intersection of red and green planes gives a 1D line perpendicular to the blue plane, which corresponds to a logical $Z$ representative. The logical $\CCZ$ from third Clifford hierarchy $\P_3$ is obtained by taking three copies of the 3D toric code and implementing physical $\CCZ$s as illustrated in Fig.~\ref{fig:cube-intersection}: for each single cube, there are six paths from one corner to the opposite corner and we apply six physical $\CCZ$s. This yields a constant-depth circuit, denoted by $U_{\CCZ}$. By the definition for any quantum state $\ket{x_1,x_2,x_3}$ on the composite system,
\begin{align}
	 U_{\CCZ} \ket{x_1,x_2,x_3} = (-1)^{T(x_1,x_2,x_3)} \ket{x_1,x_2,x_3} 
\end{align}
where 
\begin{align}
	T: \bigoplus_{i \neq j \neq k} (\mathbb{F}_2^{E_i} \otimes \mathbb{F}_2^{V_j} \otimes \mathbb{F}_2^{V_k})	\times 	\bigoplus_{i \neq j \neq k} (\mathbb{F}_2^{E_i} \otimes \mathbb{F}_2^{V_j} \otimes \mathbb{F}_2^{V_k})	
	\times 	\bigoplus_{i \neq j \neq k} (\mathbb{F}_2^{E_i} \otimes \mathbb{F}_2^{V_j} \otimes \mathbb{F}_2^{V_k})	 \rightarrow \mathbb{F}_2
\end{align}
is a trilinear function, or equivalently, a homogeneous polynomial of degree 3 on the space of physical qubits. It is determined by how we apply the physical $\CCZ$s and it turns out that $T$ is nontrivial and
\begin{align}
	T(x_1 + y_1,x_2 + y_2, x_3 + y_3) = T(x_1,x_2,x_3)
\end{align}
for any $X$ logical representatives $x_i$ and any $X$-stabilizer deformation $y_i$. We say that $T$ is invariant and preserves the cohomology classes. Therefore, $U_{\CCZ}$ is a diagonal logical gate on the code space. It happens to be a composition of logical $\CCZ$s from $\P_3$ \cite{10.21468/SciPostPhys.14.4.065,Chen_2023,Wang_2024,Breuckmann2024Cups}.

In this example, $t = 3$, $l = 1$ and $r = t - l = 2$ using our previous notation in Section~\ref{sec: hgp} and \ref{sec:constant_depth}. Let $\overline{X}_1$ be any $X$ logical on the first toric code, by definition,
\begin{align}
	U_{\CCZ} \overline{X}_1 U_{\CCZ} = \overline{X}_1 U_{\mathrm{C}Z_{2,3}}
\end{align}
where $U_{\mathrm{C}Z_{2,3}}$ is a (composition of) logical $\mathrm{C}Z$ on the second and third code copies. The circuit $U_{\CCZ}$ maps hyperplanes to hyperplanes across several sectors, which is consistent with our Assumption \ref{assmuption0} and other results in Section~\ref{sec:constant_depth}.

More generally, we can define $t$-dimensional toric codes with physical qubits corresponding to $l$-dimensional cubes. We can always build $\lfloor r/l \rfloor + 1 = \lfloor (t-l)/l \rfloor + 1 = \lfloor t/l \rfloor$ level non-Clifford gates: 
\begin{enumerate}
	\item Suppose $r = t-1$, then $l = 1$ and $\lfloor t/r \rfloor = t$. We are able to construct multi-controlled-$Z$ gates generalized from the above construction of logical $\CCZ$ \cite{Chen_2023,Breuckmann2024Cups}. When $l > 1$ and $t$ is divisible by $l$, by basic results from algebraic topology, $t/l$ products of the $l$-th cohomology of the $t$-dimensional torus yields a nontrivial multi-linear map
	\begin{align}
		T: H^l(\mathbb{T}^t, \mathbb{F}_2) \times \cdots \times H^l(\mathbb{T}^t, \mathbb{F}_2) \rightarrow H^t(\mathbb{T}^t, \mathbb{F}_2) \cong \mathbb{F}_2.
	\end{align}  
	It can be used to define the desired logical $t/l$-controlled-$Z$ gates as we define the trilinear form for $\CCZ$. 
	
	\item Suppose $t$ is not divisible by $l$. For example, if $t -1$ is divisible by $l$ instead, then we treat the $t$-dimensional toric code as the homological product of the $t-1$ dimensional toric code (with physical qubits on $l$-dimensional cubes)
	\begin{align}
		C_{l+1} \xleftarrow{\delta^l} C_l \xleftarrow{\delta^{l-1}} C_{l-1}
	\end{align}
	with an invariant form $T(x_1,...,x_{(t-1)/l})$, and the classical cyclic code $A_c: \mathbb{F}_2^n \rightarrow \mathbb{F}_2^n$ (or $H_c: \mathbb{F}_2^{E} \rightarrow \mathbb{F}_2^{V}$ for clarity):
	\begin{align}
		C_{l+1} \otimes \mathbb{F}_2^n \xleftarrow{} 
		(C_l \otimes \mathbb{F}_2^n) \oplus (C_{l+1} \otimes \mathbb{F}_2^n) \xleftarrow{}
		(C_{l-1} \otimes \mathbb{F}_2^n) \oplus (C_l \otimes \mathbb{F}_2^n) \xleftarrow{}
		C_{l-1} \otimes \mathbb{F}_2^n. 
	\end{align}
	Three consecutive terms from right-hand side defines the code. Given vectors
	\begin{align}
		(x_i' \otimes u_a, x_j \otimes v_b) \in (C_{l-1} \otimes \mathbb{F}_2^n) \oplus (C_l \otimes \mathbb{F}_2^n), 
	\end{align}
	let 
	\begin{align}
		\tilde{T}( (x_i' \otimes u_a, x_j \otimes v_b)_1,...,(x_i' \otimes u_a, x_j \otimes v_b)_{(t-1)/l} ) \vcentcolon = T(x_{j,1},...,x_{j, (t-1)/l}).
	\end{align}
	Since the cyclic code has only one nontrivial codeword---the all-ones vector---it is straightforward to check that this extension $\tilde{T}$ is invariant under $X$ stabilizers determined by $C_{l-1} \otimes \mathbb{F}_2^n$. The degree of $\tilde{T}$ is still $(t-1)/l = \lfloor t/l \rfloor$.
\end{enumerate}

In conclusion, according to our results in Section~\ref{sec:constant_depth}, this achieves the optimal Clifford hierarchy level and shows the tightness of our no-go theorem. In particular, $(l,r) = (1,t-1)$ $t$D toric code supports fault-tolerant $\mathrm{C}^{t-1}Z$ gate which attains the optimal level $\mathcal{P}_t$.

As a side remark, the existence of macroscopic energy barrier requires the absence of string-like operators which only happen for $l,r>1$, so it is not compatible with fault-tolerant $\mathcal{P}_t$ gates, which confirms a result in \cite{pastawski_fault-tolerant_2015}.

\subsection{Geometrically non-local codes}

Now we move on to the situation without geometric locality constraints.
The use of cup products to construct logical gates is generally applicable to the so-called quantum sheaf codes \cite{dinur2024cubicalcomplexes}. A quantum sheaf code is constructed by planting classical codes with certain robustness properties, as a natural generalization of the Sipser--Spielman constructions \cite{SS1996}. This family of codes encompasses all present constructions of good quantum LDPC codes \cite{pk22,QuantumTanner2022,DHLV} and many product-based constructions such as lifted product and homological product codes \cite{BalancedProduct2020,lifted}. It is recently shown that \cite{Lin2024transversal} with appropriate choices of local codes, a vast family of quantum sheaf codes could potentially support nontrivial logical $\mathrm{C}^{t-1}Z$ actions. One of the particularly interesting families is given by homological products of classical Sipser--Spielman codes, known as quantum expander codes \cite{Leverrier_2015}.

Let $\mathcal{G}_i$ be $\Delta$-regular expander graphs with edges and vertices $E_i, V_i$ and let $\mathcal{C}_i$ be classical codes with parity-check matrices $h_i \in \mathbb{F}_2^{m_i \times \Delta}$. 
A 3-dimensional quantum expander code is defined by the homological product of three classical Sipser--Spielman codes $T(\mathcal{G}_i,\mathcal{C}_i)$ based on $\mathcal{G}_i$ and $\mathcal{C}_i$:
\begin{align}
	C^3(X,\mathcal{F}) \xleftarrow{\delta^2} C^2(X,\mathcal{F}) \xleftarrow{\delta^1} C^1(X,\mathcal{F}) \xleftarrow{\delta^0} C^0(X,\mathcal{F}) 
\end{align}
where

\begin{align}
	& C^0(X,\mathcal{F}) = \mathbb{F}_2^{V_1 \times m_1} \otimes \mathbb{F}_2^{V_2 \times m_2} \otimes \mathbb{F}_2^{V_2 \times m_2}, \\
	& C^1(X,\mathcal{F}) = \bigoplus_{i \neq j \neq k} \mathbb{F}_2^{E_i} \otimes \mathbb{F}_2^{V_j \times m_j} \otimes \mathbb{F}_2^{V_k \times m_k}, \\
	& C^2(X,\mathcal{F}) = \bigoplus_{i \neq j \neq k} \mathbb{F}_2^{E_i} \otimes \mathbb{F}_2^{E_j} \otimes \mathbb{F}_2^{V_k \times m_k}, \\
	& C^3(X,\mathcal{F}) = \mathbb{F}_2^{E_1} \otimes \mathbb{F}_2^{E_2} \otimes \mathbb{F}_2^{E_3}.	
\end{align}
\comments{
\begin{align}
	& C^0(X,\mathcal{F}) = \mathbb{F}_2^{V_1 \times m_1} \otimes \mathbb{F}_2^{V_2 \times m_2} \otimes \mathbb{F}_2^{V_2 \times m_2}, \\
	& C^1(X,\mathcal{F}) = \Big( \mathbb{F}_2^{E_1} \otimes \mathbb{F}_2^{V_2 \times m_2} \otimes \mathbb{F}_2^{V_3 \times m_3} \Big) \oplus \Big( \mathbb{F}_2^{V_1 \times m_1} \otimes \mathbb{F}_2^{E_2} \otimes \mathbb{F}_2^{V_3 \times m_3} \Big) \oplus \Big( \mathbb{F}_2^{V_1 \times m_1} \otimes \mathbb{F}_2^{V_2 \times m_2} \otimes \mathbb{F}_2^{E_3} \Big), \\
	& C^2(X,\mathcal{F}) = \Big( \mathbb{F}_2^{V_1 \times m_1} \otimes \mathbb{F}_2^{E_2} \otimes \mathbb{F}_2^{E_3} \Big) \oplus \Big( \mathbb{F}_2^{E_1} \otimes \mathbb{F}_2^{V_2 \times m_2} \otimes \mathbb{F}_2^{E_3} \Big) \oplus \Big( \mathbb{F}_2^{E_1} \otimes \mathbb{F}_2^{E_2} \otimes \mathbb{F}_2^{V_3 \times m_3} \Big), \\
	& C^3(X,\mathcal{F}) = \mathbb{F}_2^{E_1} \otimes \mathbb{F}_2^{E_2} \otimes \mathbb{F}_2^{E_2}.	
\end{align}
}
It is straightforward to see that the above construction is essentially a 3-dimensional hypergraph product code but we highlight the degree of freedom from $\mathbb{F}_2^{m_i}$, called local coefficients, associated with the classical local codes $\mathcal{C}_i$ (cf. the 3D toric code in \eqref{eq:3Dtoric}). 

We still take three consecutive terms on the right-hand side to form the quantum code. A notable feature for the quantum expander code as well as their generalization---the sheaf code---is the introduction of local codes $\mathcal{C}_i$. Suppose the $\mathcal{C}_i$ satisfies the \emph{triorthogonality} condition, i.e., the sum of the components of the pointwise product of arbitrary three codewords $c_1 \odot c_2 \odot c_3$ is identical to $0$ \cite{Bravyi_2012}. With mild assumptions on $\mathcal{G}_i$, e.g., $\mathcal{G}_i$ are double covers of a certain Cayley graph (which allows for an easy construction of the so-called \emph{cubical complex} \cite{dinur2024cubicalcomplexes}), it is proven in \cite{golowich2024quantumldpccodestransversal} that the code supports a constant-depth logical $\overline{\CCZ}$ gate. This framework is also generalized to sheaf codes in \cite{Lin2024transversal}. Restricted to hypergraph product codes, for $t$-dimensional case with $r = t - 1$ and $l = 1$, we can combine techniques for performing $\mathrm{C}^{t-1}Z$ gates on topological codes and properties of the local codes, e.g., the multi-orthogonality, to design the global logical gate. It also is expected to work for arbitrary $l+r = t$ (with $l \leq r$) as explained at the end of Section~\ref{sub:topo}, demonstrating an optimality analogous to that for higher-dimensional toric codes. 

\section{Conclusions and outlook}\label{sec: conclusions and outlook}

Motivated by the pressing importance of systematically understanding how information protection interacts with logical computation in the quest for practical quantum computing, we set forth to derive  fundamental no-go theorems for fault-tolerant logical gates originated from certain QEC code structures. 
As a starting point, in this work we established the general algebraic constraints for gates and thoroughly examined the case of hypergraph product codes, which constitute a particularly elegant yet rich framework for constructing QEC (especially qLDPC) codes with flexible code parameters and many highly appealing features for quantum computing.
Our results reveal fundamental restrictions on the logical Clifford hierarchy levels that can be reached with fault-tolerant gates under a variety of conditions,
significantly generalizing the Bravyi--K\"{o}nig argument for topological codes.

A key point is the relaxation of geometric connectivity constraints---the code is only assumed to possess an algebraic product structure and not any geometric feature, which extends greatly beyond topological codes.
Despite requiring more  intricate technical treatment, our generalized no-go theorems exhibit closely relevant forms with the geometric ones.  This underscores the fundamentally algebraic nature of the logical gate problem, in contrast to other central code properties including rate and distance that strongly depend on geometry.

From the perspective of fault tolerance, our results further highlight the inherent difficulty of non-Clifford gates even when geometric constraints that are considered a key limitation of many current experimental technologies are removed, calling for  further attention on the study of logical gate properties of codes and comparisons with other possible strategies for logical gate implementation. 
This deepens our understanding of the tension between quantum information protection and processing which lies at the heart of the potential value of quantum computing. Although on the other hand, an encouraging message from the optimal yes-go constructions, especially given the rapid progress for experimental platforms with good connectivity/reconfigurability (such as neutral atom arrays and ion traps), is that codes with simple structures that simultaneously carry favorable parameters, gates and decoding properties beyond what geometric codes can offer are practically promising.

Several directions and open problems are worth further exploration.
First, we have focused on the arguably neatest hypergraph product codes in this work. Many insights and methods here can provide a foundation for or be extended to even more general code structures including lifted product codes and sheaf codes,
which will be explicated in forthcoming works. Eventually, we seek to establish a comprehensive understanding of the joint interplay between code structures and logical gates as well as other code properties.
Second, higher dimensional chain complexes underlie not only the logical gate problem, but also various crucial code properties including local testability, single-shot decoding, etc. 
We expect further investigations both conceptual and technical connections among them to be fruitful.
Third, deeper understanding of fault-tolerant logical gates on qLDPC codes are anticipated to have interesting implications for the quantum PCP conjecture \cite{Anshu_2024,Lin2024transversal}. Lastly, while the current of study on logical gates builds upon the stabilizer code formalism, non-stabilizer codes are expected to provide a novel lens into logical gates, which we consider to be a rich and important avenue.

\section*{Acknowledgments}
We thank Daniel Gottesman, Tomas Jochym-O'Connor, Colin La, Samuel Tan, Guangqi Zhao, Alexander Barg for valuable discussion and feedback.  Part of this work is done while XF is visiting YMSC, Tsinghua University. ZL and ZWL are supported in part by a startup funding from YMSC, Dushi Program, and NSFC under Grant No.~12475023.

\bibliography{biblio}

\end{document}